\documentclass[preprint,12pt]{elsarticle}

\usepackage[utf8]{inputenc}
\usepackage{balance}
\usepackage{mathrsfs}
\usepackage{multirow}
\usepackage{multicol}
\usepackage{amsmath}
\usepackage{amsthm}
\usepackage[vlined,ruled,linesnumbered,boxed]{algorithm2e}
\usepackage[OT2,OT1]{fontenc}
\usepackage{amssymb}
\usepackage{appendix}
\usepackage{color}
\usepackage{amstext}
\usepackage{hyperref}
\usepackage{array}
\usepackage{graphicx}
\newtheorem{theorem}{Theorem}[section]
\newtheorem{proposition}[theorem]{Proposition}
\newtheorem{lemma}[theorem]{Lemma}
\newtheorem{corollary}[theorem]{Corollary}
\newtheorem{example}[theorem]{Example}
\newtheorem{remark}[theorem]{Remark}
\newtheorem{assumption}[theorem]{Assumption}
\newtheorem{definition}[theorem]{Definition}
\newtheorem{observation}[theorem]{Observation}

\journal{the Journal of Algebra}

\begin{document}

\begin{frontmatter}


\title{A Vergleichsstellensatz of Strassen's Type for a Noncommutative Preordered Semialgebra through the Semialgebra of its Fractions}
\author[a1]{Tao Zheng}
\author[a1,a2]{Lihong Zhi}
\address[a1]{{Key Laboratory of Mathematics Mechanization, Academy of Mathematics and Systems Science, Chinese Academy of Sciences},
            {Beijing},
            {100190}, 
            {China}}
\address[a2]{{University of Chinese Academy of Sciences},
            {Beijing},
            {100049}, 
            {China}}
\begin{abstract}
Preordered semialgebras and semirings are two algebraic structures frequently occurring in real algebraic geometry. They have many interesting and promising applications in the fields of probability theory, theoretical computer science, quantum information theory, \emph{etc.}. Strassen's Vergleichsstellensatz and its generalized versions, analogs of those well-known Positivstellensätze, play important roles in these applications. While these Vergleichsstellensätze accept only a commutative setting (for the semirings in question), we prove in this paper a noncommutative version of one of the generalized Vergleichsstellensätze proposed by Fritz [\emph{Comm. Algebra}, 49 (2) (2021), 482-499]. The most crucial step in our proof is to define the semialgebra of the fractions of a noncommutative semialgebra, which generalizes (at least some of) the definitions in the literature. Our new theorem characterizes the relaxed preorder on a noncommutative semialgebra induced by all monotone homomorphisms to $\mathbb{R}_+$ via three other equivalent conditions on the semialgebra of its fractions.
\end{abstract}

\begin{keyword}
semiring \sep Strassen's theorem \sep Positivstellensatz
\end{keyword}
\end{frontmatter}

\section{Introduction}

A (noncommutative) \emph{semiring} $S$ is a set together with binary operations 
\[+,*:S\times S\rightarrow S,\]respectively called \emph{addition} and \emph{multiplication}. Elements $0_S,1_S\in S$ such that $(S,+,0_S)$ is a commutative monoid, $(S,*,1_S)$ is a noncommutative monoid and such that the multiplication distributes over the addition \cite{golan2013semirings}. A (noncommutative) \emph{semialgebra} $S$ is a (noncommutative) semiring together with a (nonnegative) scalar multiplication 
\[\mathbb{R}_+\times S\rightarrow S,(r,x)\mapsto r\cdot x\] that is a commutative monoid homomorphism from $(\mathbb{R}_+,+)$ or $(S,+)$ to $(S,+)$ in the first or the second argument respectively, and satisfies the following addition laws: \emph{i})\; $1\cdot x=x$,\, \emph{ii})\; $(r\cdot x)*(s\cdot y)=(rs)\cdot(x*y)$ \cite[Defintion 3.5]{algeP2}.

A semialgebra $S$ is \emph{zero-sum-free} if for any $a,b\in S$, $a+b=0_S$ implies $a=b=0_S$. It is \emph{zero-divisor-free} if for any $a,b\in S$, $a*b=0_S$ implies $a=0_S$ or $b=0_S$. In a zero-divisor-free semialgebra $S$, $r\cdot a=0_S$ implies $r=0$ or $a=0_S$ for any $r\in\mathbb{R}_+$ and $a\in S$: $r\cdot a=0_S$ $\Rightarrow$ $(r\cdot1_S)* a=0_S$ $\Rightarrow$ $r\cdot1_S=0_S$ or $a=0_S$. If $a\neq0_S$ and $r\neq0$, then $1_S=r^{-1}\cdot(r\cdot1_S)=r^{-1}\cdot0_S=0_S$. Hence $a=a*1_S=a*0_S=0_S$, which is a contradiction.

A \emph{preorder relation} (or, \emph{preorder}) $\leq$ on a set $X$ is a binary relation that is both reflexive and transitive (\cite{golan2013semirings}, page 119). Set $x,y\in X$, we sometimes write ``$x\geq y$" instead of ``$y\leq x$", and write ``$x<y$" or ``$y>x$" instead of the condition ``$x\leq y\text{ and }x\neq y$" throughout the paper. A preorder relation $\leq$ on $X$ is considered \emph{trivial} if $x\leq y$ holds for any $x,y\in X$.

A \emph{preordered semiring} \cite[Definition 2.2]{fritz2021generalization} is a semiring with a preorder relation $\leq$ such that for all $a,x,y$ in the semiring,
\begin{equation}\label{preoc}
x\leq y\text{ implies }\left\{
\begin{array}{rcl}
a+x&\leq&a+y,\\
a*x&\leq&a*y,\\
x*a&\leq&y*a.
\end{array}
\right.
\end{equation}
Similarly, a \emph{preordered semialgebra} is a semialgebra with a preorder relation $\leq$ such that $x\leq y$ implies the inequalities in (\ref{preoc}). Note that it also implies the inequality 
\begin{equation}\label{scalar}
r\cdot x\leq r\cdot y
\end{equation}
for any $r\in\mathbb{R}_+$, since $r\cdot x=(r\cdot1_S)*x\leq(r\cdot1_S)*y=r\cdot y$. 

Let $E$ and $K$ be two semialgebras. A map $f$ from $E$ to $K$ is a \emph{semialgebra homomorphism} (or, simply \emph{homomorphism}) if $f(0_E)=0_K,f(1_E)=1_K$ and for any $x,y\in E$ and any $r\in\mathbb{R}_+$, the following equations hold
\[f(x+y)=f(x)+f(y),\;f(x*y)=f(x)*f(y)\text{ and }f(r\cdot x)=r\cdot f(x).\]If $\leq_{_E}$ and $\leq_{_K}$ are two preorder relations on  $E$ and $K$ respectively, then a map $f$ from $E$ to $K$ is \emph{monotone} (\emph{w.r.t.} $\leq_{_E}$ and $\leq_{_K}$) if for any $x,y\in E$, $x\leq_{_E} y$ implies $f(x)\leq_{_K} f(y)$.

The following definition of a \emph{power universal} element in a noncommutative preordered semialgebra is a natural generalization of   \cite[Definition 2.8]{fritz2021generalization}.
\begin{definition}\label{poweruni}
Let $S$ be a noncommutative preordered semialgebra with $1_S\geq0_S$. An element $u\in S$ with $u\geq1_S$ is \emph{power universal} $($\emph{w.r.t.} the preorder $\leq)$ if for every nonzero $x\in S$ there is a number  $k\in\mathbb{N}$ such that
\begin{equation}\label{PU}
u^k*x\geq1_S,\;\;x*u^k\geq1_S\;\text{ and }\;u^k\geq x.
\end{equation}
\end{definition}

 Recently, Strassen's  separation theorem \cite{strassen1988asymptotic,zuiddam2019asymptotic} for preordered semirings (called ``Strassen's Vergleichsstellensatz" \cite{algeP1}) has been generalized tremendously by Fritz and Vrana \cite{algeP2,fritz2021generalization,algeP1,vrana2021generalization} due to its varies applications to real algebraic geometry, probability theory, theoretical computer science, and quantum information theory \cite{fritz2021generalization,zuiddam2019asymptotic,li2020quantum,perry2021semiring,bunth2021asymptotic,robere2022amortized,1stapp}. Strassen's Vergleichsstellensatz and its generalizations are analogs of those Positivstellensätze in real algebraic geometry focusing mainly on ordered rings and fields. Among them, Fritz's results \cite{fritz2021generalization,algeP1} recover the classical Positivstellensatz of Krivine-Kadison-Dubois, which suggests that the theory of semirings and semialgebras can provide us with new insights into real algebraic geometry.

While Strassen's Vergleichsstellensatz and its generalizations accept only the commutative setting for the semirings under consideration, we prove in this paper a noncommutative version of one of Fritz's generalized Vergleichsstellensätze in \cite{fritz2021generalization}. As in the commutative case, the main assumption on the preordered semialgebra $S$ considered in our theorem is:
\begin{assumption}\label{power}
The inequality $``\,1_S\geq0_S"$ holds in the preordered semialgebra $S$, and there is a {power universal} element $u\in S$.
\end{assumption}
\noindent{}The assumption of the existence of a power universal element is similar to the Archimedean condition in the traditional real algebraic geometry.  
Moreover,  the semialgebra $S$ considered in this paper also   satisfies the following assumption: 
\begin{assumption}\label{assfree}
The semialgebra $S$ is both {zero-sum-free} and {zero-divisor-free},  and the inequality $``\,1_S\neq0_S"$ holds. 
\end{assumption}
Then our main result  can be stated as follows:
\begin{theorem}\label{main}
Let $(S,\leq)$ be a preordered semialgebra satisfying Assumptions \ref{power} and \ref{assfree}, with a power universal element $u\in S$. Then, for every nonzero $x,y\in S$, the following are equivalent:

\noindent$($a$)$ $f(x)\geq f(y)$ for every monotone semialgebra homomorphism $f:S\rightarrow\mathbb{R}_+$.

\noindent$($b$)$ For every real number $\epsilon>0$, there is a finite number $m\in\mathbb{N}$ such that 
\begin{equation}
\overline{x}+\sum_{j=0}^m\epsilon^{j+1}\cdot\overline{u}^j\succcurlyeq\overline{y}.
\end{equation}

\noindent$($c$)$ For every  real number $r\in\mathbb{R}_+$ and every real number $\epsilon>0$, there is a polynomial $p\in\mathbb{Q}_+[X]$ such that $p(r)\leq\epsilon$ and 
\[
\overline{x}+p(\overline{u})\succcurlyeq\overline{y}.
\]

\noindent$($d$)$ For every real number  $r\in\mathbb{R}_+$ and every real number $\epsilon>0$, there is a polynomial $p\in\mathbb{Q}_+[X]$ such that $p(r)\leq1+\epsilon$ and 
\begin{equation}
p(\overline{u})*\overline{x}\succcurlyeq\overline{y}.
\end{equation}
\end{theorem}
\noindent{}The elements $\overline{x},\overline{y}$ and $\overline{u}$ above are the images of $x,y$ and $u$ for the canonical map (Definition \ref{canonical}) from $S$ to the semialgebra of its fractions, and the preorder ``$\preccurlyeq$'' is a preorder on the semialgebra of the fractions which is directly derived from the original preorder ``$\leq$" on $S$ (Definition  \ref{def:derivedPreorder}).

\begin{example}\label{exp:classic}
Let $\mathbb{R}_+\hspace{-1mm}\left\langle x_1,\ldots,x_n\right\rangle'$ be the set of all noncommutative polynomials in $n$ variables with coefficients in $\mathbb{R}_+$ whose constant terms are nonzero. Then,  the semialgebra $S=\{0\}\cup\mathbb{R}_+\hspace{-1mm}\left\langle x_1,\ldots,x_n\right\rangle'$ equipped with the coefficient-wise preorder is one of the simplest examples of  preordered semialgebras. 
It satisfies Assumptions \ref{power}--\ref{assfree} with the element $u=2+\sum_{i=1}^n2\cdot x_i$ being power universal. Theorem \ref{main} can therefore be applied.
\end{example}

The main technical difficulty in proving Theorem \ref{main} is to define properly the semialgebra of the fractions of a noncommutative semialgebra: In the commutative case, every fraction of  elements in a commutative semialgebra/semiring can be written in the simple form ``$\frac{a}{b}$" for two elements $a,b$ in the semialgebra/semiring. However, ``fractions" in the noncommutative case can be of more complicated forms, \emph{e.g.}, ``$d^{-1}*g*h^{-1}$" and ``$(d+g^{-1})^{-1}$", which can no longer be written in the form ``$\frac{a}{b}$." Our definition for the semialgebra of the fractions of a noncommutative semialgebra in Subsection \ref{subsec:frac} generalizes the definition in the literature, and we also show  that it coincides with the usual definition in the commutative case (Proposition \ref{isom}). Now that theories on commutative semirings and semialgebras help us better understand commutative real algebraic geometry, it is reasonable to expect the same thing in the noncommutative case. Moreover, we hope the theory of noncommutative semirings and semialgebras (including our new theorem) could be applied to  quantum information theory and other related  areas.  

The main difference between Theorem \ref{main} and  Fritz's Vergleichsstellensatz (\cite[Theorem 2.12]{fritz2021generalization}) is that, in the commutative case, the corresponding inequalities in conditions (b) - (d) can be rewritten as inequalities in the original semialgebra $S$ (Proposition \ref{preorderisom}), while in Theorem \ref{main}, one may find it non-trivial to do the same thing. 

In the next section, we define the semialgebra of the fractions of a noncommutative semialgebra that satisfies Assumption \ref{assfree}. This definition is the cornerstone of the main theorem of our paper. Section \ref{sec:derivedPreorder} indicates how a preorder relation on a semialgebra derives another preorder relation on the semialgebra of its fractions.  Section \ref{sec:extend} is devoted to interpreting how an $\mathbb{R}_+$-valued semialgebra homomorphism can be extended from a semialgebra to the semialgebra of its fractions. In Section \ref{sec:V}, we explain in detail how an $\mathbb{R}$-linear space can be constructed from a (noncommutative) semialgebra. This is mentioned in \cite{fritz2021generalization}, but we include it for rigorousness and completeness. Finally, in  Section \ref{sec:mainTHM}, we present  the proof of Theorem \ref{main}. 

\section{The Semialgebra of the Fractions of a Noncommutative Semialgebra}\label{sec:frac}
Throughout the paper, the letter ``$S$" stands for a noncommutative semialgebra that satisfies Assumption \ref{assfree}.

The following example introduces a simple method to  define the semialgebra of nonnegative rational numbers $\mathbb{Q}_+$ from the semialgebra of nature numbers $\mathbb{N}$. Then we define the semialgebra of the fractions of $S$ similarly.
\begin{example}\label{exp:rational}
To define $\mathbb{Q}_+$ from $\mathbb{N}$, there are 3 steps:

i) Define the set of formal expressions $U=\{n\oslash m\,|\,n,m\in\mathbb{N}\}$ with $\oslash$ the formal division.

ii) There are obviously some ``illegal" expressions like $3\oslash0$ and $0\oslash0$, but we only care about the set of ``legal" expressions $W=\{n\oslash m\,|\,n,m\in\mathbb{N}, m>0\}$.

iii) Since some expressions, \emph{e.g.}, $2\oslash6$ and $3\oslash9$, stand for the same rational number, we need an equivalence relation $R$ on $W$ telling whether two expressions ``equal". Here is a simple way to define $R$: $n\oslash m\stackrel{\,\hspace{-0.3mm}_R}{\sim}i\oslash j$ if and only if  $nj=mi$. We then define $\mathbb{Q}_+=W/R$, which can be regarded as the semialgebra of the fractions of $\mathbb{N}$.
\end{example}

We  can define the semialgebra of the fractions of $S$ similarly.
 

\subsection{Defining the set of formal rational expressions of the elements in $S$}\label{subsec:formalexpression}

The first step is to construct the set of all formal expressions of fractions (as the set $U$ in Example \ref{exp:rational}). Since $S$ is noncommutative, the ``fractions" of $S$ may have more complicated forms than just $a/b$. For instance, we should allow expressions like $(2\cdot a*b+c^{-1})^{-1}$, which contains multiplication, scale multiplication, addition, and inversion. Thus, we define $U$ to be the set of all finite formal expressions in the elements in $S$ with the formal addition ``$\oplus$", the formal multiplication ``$\circledast$," the formal scalar multiplication ``$\odot$" and the formal inversion ``$\;^{-1}$" ($\oplus$, $\circledast$, and $\odot$ are binary operations while $\;^{-1}$ is unary). To be precise, we have 
\begin{definition}\label{def:U}
The set $U(S)$ $($or simply $U$, if it won't cause any ambiguity$)$ of \emph{formal rational expressions} consisting of finitely many formal operations and elements in a semialgebra $S$, refers to the set determined exactly by the following rules:

i) $S\subset U$,

ii) $a\circledast b\in U$ if and only if both $a,b\in U$,

iii) $a\oplus b\in U$  if and only if both $a,b\in U$,

iv) $r\odot a\in U$  if and only if $r\in\mathbb{R}_+$ and $a\in U$,

v) $a^{-1}\in U$  if and only if $a\in U$.
\end{definition}

There is no such element as $a\oplus b\oplus c$ in $U$ since $\oplus$ is a binary operation. But there are elements of $U$ of the form $(a\oplus b)\oplus c$ or $a\oplus (b\oplus c)$. By Definition \ref{def:U}, every element of $U$ either contains no formal operations (\emph{i.e.}, it is in $S$) or is of exactly one of the following forms: $r\odot a, a\circledast b, a\oplus b, a^{-1}$. To illustrate, set $s_i\in S$, then these are respectively some elements of $U$ of those forms:

\[2\odot(s_1\oplus s_2)^{-1},\;\;0_S\circledast s_1,\;\;(0_S\circledast s_1)\oplus (0_{\mathbb{R}}\odot s_2^{-1}),\;\;(0_S\circledast s_1)^{-1}.\]
It is clear that the last expression is ``illegal" since it is the inverse of  an expression that ``equals" to zero, which makes no sense. Therefore, the set of ``legal" expressions in $U$ are those ones such that whenever they contain a sub-expression of the form $a^{-1}$, $a$ is not an expression ``equaling" zero. On the other hand, there are many expressions in $U$ ``equaling" zero, as the second and the third ones shown above. The set of all ``legal" expressions (denoted by $W$) is not as clear as in the commutative case in Example \ref{exp:rational}, since the set (denoted by $\mathcal{O}$) of the expressions ``equaling" zero is more complicated. These two sets have to be defined together because they ``tangle" with each other: for any $a\in W$, $a^{-1}$ should be in $W$  if and only if  $a\notin\mathcal{O}$, and for any $a\in W$, $0_S\circledast a$ should be in $\mathcal{O}$. The precise definition of them is as follows:
\begin{definition}\label{def:WO}
The sets of \emph{legal}  and \emph{null} formal rational expressions of the elements in a semialgebra $S$, denoted by $W(S)$ $($or $W)$ and $\mathcal{O}(S)$ $($or $\mathcal{O})$ respectively, refer to the subsets of $U$ which are determined exactly by the following rules:

i) $S\subset W$ and $\mathcal{O}\cap S=\{0_S\}$;\vspace{1mm}

\noindent{}for all $a, b\in U$,

ii) $a\oplus b\in W$  if and only if $a,b\in W$, and 

$\quad\; a\oplus b\in \mathcal{O}$  if and only if $a,b\in \mathcal{O}$;

iii) $a\circledast b\in W$  if and only if $a, b\in W$, and 

$\quad\;\; a\circledast b\in\mathcal{O}$  if and only if one of $a, b$ is in $\mathcal{O}$ but the other is in $W$;\vspace{1mm}

\noindent{}for all $a\in U$ and $r\in{\mathbb{R}_+}$,

iv) $r\odot a\in W$  if and only if $a\in W$, and 

$\quad\;\, r\odot a\in \mathcal{O}$  if and only if $a\in\mathcal{O}$ or $(r=0_\mathbb{R}\wedge a\in W)$;\vspace{1mm}

\noindent{}for all $a\in U$,

v) $a^{-1}\notin\mathcal{O}$, and

$\quad\; a^{-1}\in W$ if and only if $a\in W\backslash\mathcal{O}$.
\end{definition}
All these rules are natural and easy to understand: Rule i) means elements of $S$ are ``legal" expressions, and $0_S$ is the only element in $S$ that ``equals" zero. The first conditions of Rules ii) -- iv) mean that the ``legal" expressions are closed concerning  the formal operations $\oplus, \circledast$ and $\odot$ and that only ``legal" expressions can generate other ``legal" expressions. The second condition in Rule v) means that only the inverses of ``legal" expressions not ``equaling" zero are ``legal" inverses. The second condition in Rule ii) says the set of ``zeros" is closed for the formal addition, and the zero-sum-free property holds. The second condition in Rules iii)  -- iv) means zero times anything ``legal" is again zero, and that the zero-divisor-free property holds. The condition $a^{-1}\notin\mathcal{O}$ in Rule v) is quite natural since anything of the form $a^{-1}$ cannot be zero.

Definition \ref{def:WO} does uniquely characterize two subsets of $U$ since there is an algorithm deciding whether a given expression $a\in U$ is in $W$ (or $\mathcal{O}$) or not. Some examples are sufficient for the readers to understand that:
\begin{example}\label{exp:algoWO}
Let $s_i$ be some nonzero elements of $S$. Then we have 
\[
\begin{array}{cl}
     &(s_1\oplus s_2)^{-1}\in W  \\
     \Leftrightarrow&s_1\oplus s_2\in W\backslash\mathcal{O} \\
     \Leftrightarrow&s_1\oplus s_2\in W \wedge s_1\oplus s_2\notin\mathcal{O}\\
     \Leftrightarrow& s_1, s_2\in W \wedge (s_1\notin\mathcal{O}\vee s_2\notin\mathcal{O})\\
     \Leftrightarrow& \emph{\textbf{True}}.
\end{array}
\]
The reasoning above uses only rules in Definition \ref{def:WO} time after time. Whenever a rule is used, an operation in the expression $(s_1\oplus s_2)^{-1}$ is reduced. When there are no operations, Rule i) in Definition \ref{def:WO} helps decide true or false. Here is another example:
\[
\begin{array}{cl}
     & (0_S\circledast s_1)^{-1}\in W \\
    \Leftrightarrow& 0_S\circledast s_1\in W\backslash\mathcal{O}\\
    \Leftrightarrow& 0_S\circledast s_1\in W \wedge 0_S\circledast s_1\notin \mathcal{O}\\
    \Leftrightarrow&(0_S\circledast s_1\in W)\wedge 0_S\notin\mathcal{O} \wedge s_1\notin\mathcal{O}\\
    \Leftrightarrow&\emph{\textbf{False}},
\end{array}
\]
which means $(0_S\circledast s_1)^{-1}$ is ``illegal". We can also decide whether an element of $U$ is in $\mathcal{O}$. For instance, we have $(s_1\oplus s_2)^{-1}\notin\mathcal{O}$ according to Rule v) in Definition \ref{def:WO}, and
\[
\begin{array}{cl}
     &s_1\oplus s_2\in\mathcal{O}\\
    \Leftrightarrow& s_1, s_2\in\mathcal{O}\\
    \Leftrightarrow&\emph{\textbf{False}}
\end{array},
\text{\;\;\;\;\;\;and\;\;\;\;\;\;}
\begin{array}{cl}
     &0_S\circledast s_1\in\mathcal{O}\\
    \Leftrightarrow& 0_S\text{ or }s_2\in\mathcal{O}\\
    \Leftrightarrow&\emph{\textbf{True}}
\end{array}.
\]
\end{example}
The following proposition shows that all expressions ``equaling" zero are ``legal", which is important.
\begin{proposition}
 The set $\mathcal{O}$ is a subset of $ W$.
\end{proposition}
\begin{proof}
    It suffices to show that $\forall i\in\mathbb{N},  \forall a\in\mathcal{O}, ``a$ has $i$ operations in it" implies ``$a\in W$".  But this can be proved inductively (for $i$)  using Definition \ref{def:WO}.
\end{proof}

The proposition below gives another definition for the sets $W$ and $\mathcal{O}$.
\begin{proposition}\label{prop:anotherDEFofWO}
Set $W_0=S$ and $\mathcal{O}_0=\{0_S\}$. Define the recurrence sequences $\{W_i\}$ and $\{\mathcal{O}_i\}$ as follows:
\begin{equation}\label{W}
    W_{i+1}=W_i\cup(W_i\oplus W_i)\cup(W_i\circledast W_i)\cup(\mathbb{R}_{+}\odot W_i)\cup(W_i\backslash \mathcal{O}_i)^{-1},
\end{equation}
\begin{equation}\label{O}
    \mathcal{O}_{i+1}=\mathcal{O}_{i}\cup(\mathcal{O}_{i}\oplus\mathcal{O}_{i})\cup(\mathcal{O}_{i}\circledast W_i)\cup(W_i\circledast\mathcal{O}_{i})\cup(\{0_{\mathbb{R}}\}\odot W_i)\cup(\mathbb{R}_+\odot\mathcal{O}_{i}),
\end{equation} 
where $W_i\oplus W_i$ means $\{a\oplus b\,|\,a, b\in W_i\}$, $(W_i\backslash \mathcal{O}_i)^{-1}$ means $\{a^{-1}\,|\,a\in W_i\backslash \mathcal{O}_i\}$ and the meanings of the other sets in the expressions are similar. Then we have $W=\cup_{i=0}^\infty W_i$ and $\mathcal{O}=\cup_{i=0}^\infty \mathcal{O}_i$ for the sets $W$ and $\mathcal{O}$ in Definition \ref{def:WO}.
\end{proposition}
\begin{proof}
Set $\hat{W}=\cup_{i=0}^\infty W_i$ and $\hat{\mathcal{O}}=\cup_{i=0}^\infty \mathcal{O}_i$. It suffices to show that they satisfy all the rules in Definition \ref{def:WO}, since the pair $(W, S)$ of subsets of $U$ satisfying those rules is unique. It is clear that $\hat{W}$ and $\hat{\mathcal{O}}$ satisfy all those rules, excluding the second one in rule v).  To show that they also satisfy that rule, it suffices to have $W_i\backslash\hat{\mathcal{O}}=W_i\backslash\mathcal{O}_i$ for any $i\in\mathbb{N}$. This is a direct corollary of Lemma \ref{lemma:WgangO}.
\end{proof}
\begin{lemma}\label{lemma:WgangO}
For any $i\in\mathbb{N}$ and the sets $W_i$, $\mathcal{O}_i$ and $\hat{\mathcal{O}}=\cup_{i=0}^\infty \mathcal{O}_i$ defined in Proposition \ref{prop:anotherDEFofWO}, we have $W_i\cap\hat{\mathcal{O}}=\mathcal{O}_i$.
\end{lemma}
\begin{proof}
It is clear that $W_i\cap\hat{\mathcal{O}}\supset\mathcal{O}_i$ since $\mathcal{O}_i\subset W_i$ and $\mathcal{O}_i\subset\hat{\mathcal{O}}$. By simply reasoning inductively on the index $i$,  one shows the opposite containment, which we omit the details. The trick is, for any $w\in W_k\cap\hat{\mathcal{O}}$, it is of one of the forms: $a\oplus b, a\circledast b$ and $ r\odot a$. In every case, we have $a, b\in W_{k-1}$. Using the rules in Definition \ref{def:WO} for the set $\hat{\mathcal{O}}$, one finds that  the induction step works.
\end{proof}
In the rest of the paper, some proofs will be carried out based on the definition of the sets $W$ and $\mathcal{O}$ given in Proposition \ref{prop:anotherDEFofWO}.

\subsection{Constructing an equivalence relation on the set of ``legal" expressions}
As in Example \ref{exp:rational},  we define in this subsection an equivalence relation $R(S)$ (or written as $R$) on the set $W$ so that the quotient $W/R$ becomes a semialgebra of the fractions of $S$ with every nonzero element invertible. This is why it is called the \emph{semialgebra-oriented} equivalence relation (Definition \ref{def:R}). The principal is to ensure that  $R$ includes only necessary relations, such that $R$ is as small as possible, and the semialgebra of fractions $W/R$ can then be as general as possible.

On the one hand, there are no semialgebra laws in the set of ``legal" expressions $W$, \emph{e.g.},  $w_1=(a\oplus b)\oplus c$ and   $w_2=a\oplus (b\oplus c)$ are not the same. On the other hand, from Example \ref{exp:rational}, we know that $(a, b)\in R$ means $a=b$ in the semialgebra of fractions $W/R$. Therefore, we will have $w_1=w_2$ (and thus the associative law of addition) in $W/R$  if and only if we include the pair $(w_1,w_2)$ into the equivalence relation $R$. Since we expect $W/R$ to be a semialgebra, all pairs (like $(w_1, w_2)$ mentioned above) which stand for any semialgebra axioms should necessarily be included in $R$. These pairs are in the sets listed below:
\begin{equation}\label{semialgebraLaws}
\left.
\begin{array}{rcl}
    A_1&=&\{((a\oplus b)\oplus c,a\oplus (b\oplus c)),((a\circledast b)\circledast c,a\circledast (b\circledast c)),\\
        &&\textcolor{white}{h}(a\circledast(b\oplus c),(a\circledast b)\oplus(a\circledast c)),((b\oplus c)\circledast a,(b\circledast a)\oplus(c\circledast a)),\\\vspace{2mm}
    &&\textcolor{white}{h}(a\oplus b,b\oplus a),(0_S\oplus a,a),(a\oplus0_S,a)\,|\,a,b,c\in W\},\\\vspace{2mm}
    A_2&=&\{(1_S\circledast a,a),(a\circledast1_S,a),(0_S\circledast a,0_S),(a\circledast0_S,0_S)\,|\,a\in W\},\\
    A_3&=&\{(r\odot(a\oplus b),(r\odot a)\oplus( r\odot b)),((r+_{_{\mathbb{R}}}t)\odot a,(r\odot a)\oplus (t\odot a)),\\
    &&\textcolor{white}{h}(r\odot(a\circledast b),(r\odot a)\circledast b),(r\odot(a\circledast b),a\circledast (r\odot b)),\\\vspace{2mm}
    &&\textcolor{white}{h}((rt)\odot a,r\odot(t\odot a)),(1_
    \mathbb{R}\odot a,a),(0_
    \mathbb{R}\odot a,0_S)\,|\,r,t\in\mathbb{R}_+,a,b\in W\}.
\end{array}
\right.
\end{equation}
Similarly, since $(a\circledast b)^{-1}$ and $b^{-1}\circledast a^{-1}$ are different expressions in $W$, we need to include into $R$ those axioms involving the  inversion:
\begin{equation}\label{inversionLaws}
\left.
\begin{array}{rcl}
   A_4&=&\{(a^{-1}\circledast a,1_S),(a\circledast a^{-1},1_S),((a\circledast b)^{-1},b^{-1}\circledast a^{-1}),(1_S^{-1},1_S),\\\vspace{2mm}
    &&\textcolor{white}{h}(s\odot(a^{-1}),(\frac{1}{s}\odot a)^{-1})\,|\,a,b\in W\backslash\mathcal{O},s\in\mathbb{R}_{>0}\}.
\end{array}
\right.
\end{equation}
Besides, we also expect that the identities in $S$ remain valid in $W/R$. For instance,  suppose $s_1=x_1*x_2+2, s_2=x_1*x_2*x_1*x_2+2$ and $s_3=2x_1*x_2+1$ are elements in the semialgebra $S$ defined in Example \ref{exp:classic}, then $s_1^2=s_2+2s_3$ is an identity in $S$. This identity should also be true in the semialgebra of the fractions of $S$. To explain that explicitly,  we define \[Q=\{w\in W\,|\,w\text{ does not contain the formal inversion}\},\] then any $w\in Q$ can be evaluated in the semialgebra $S$: one takes off all the ``circles" of the symbols $\oplus,\circledast$ and $\odot$, 
turning them into $+,*$ and $\cdot$ (the operations in $S$) respectively, whenever they occur in the expression $w$. In this way, the expression $w$ results in an element of $S$, denoted by $w_S$.  To illustrate, let $s_1$ and $s_2$ be as above, then $w=s_2\oplus(2\cdot s_3)$ is an expression without the formal inversion and $w_S=s_2+2s_3$, which is an element in $S$. Now the condition ``identities in $S$ remains valid in $W/R$" require us to include into $R$ the following set of pairs:
\begin{equation}\label{identities}
\left.
\begin{array}{rcl}
   A_5 &=&\{(a,b)\in Q\times Q\,|\,a_S=b_S\}.
\end{array}
\right.
\end{equation}
Finally, the expressions in the set $\mathcal{O}$ defined in Definition \ref{def:WO} should be regarded as zero. Thus we need to include into $R$ the pairs in the set below:
\begin{equation}\label{O=0}
    A_6=\{(a, 0_S)\,|\,a\in\mathcal{O}\}=\mathcal{O}\times\{0_S\}.
\end{equation}

From above, the pairs in the set $\cup_{i=1}^6A_i$ should be included in $R$. However, to define operations in the semialgebra of fractions $W/R$, $R$ needs to be closed under formal operations $\oplus, \circledast, \odot$  and $\,^{-1}$.  This means 

i) $(a, b), (c, d)\in R\Rightarrow(a\oplus c, b\oplus d),  (a\circledast c, b\circledast d)\in R$,

ii) $(a, b)\in R \Rightarrow (r\odot a, r\odot b)\in R$, and

iii) $\big((a, b)\in R \wedge a\notin\mathcal{O} \wedge b\notin\mathcal{O}\big) \Rightarrow  (a^{-1}, b^{-1})\in R$

\noindent{}for any $a, b, c, d\in W$ and any $r\in\mathbb{R}_+$. To see why the set $R$ has to meet these conditions, suppose $\overline{a}, \overline{b}\in W/R$ are the equivalence classes  of $a, b\in W$, respectively. Then it is natural to define the addition in $W/R$ by
\[\overline{a}+\overline{b}=\overline{a\oplus b}.\]Now ``the relation $R$ is closed under the formal addition $\oplus$" ensures that this addition is well-defined. Otherwise, it seems impossible to define an addition in $W/R$.

Note that  if $R_1, R_2\subset W\times W$ are two equivalence relations on $W$ which are both closed under all the formal operations, then so is their intersection $R_1\cap R_2$.  Therefore we have the following definition
\begin{definition}\label{def:R}
The \emph{semialgebra-oriented} equivalence relation $R(S)\subset W\times W$  on the set $W$ of legal formal rational expressions in elements of a semialgebra $S$ is defined to be the minimal subset of $W\times W$ such that (1) it contains $\cup_{i=1}^6A_i$, that (2) it is an equivalence relation and that (3) it is closed under the formal operations $\oplus, \circledast, \odot$  and $\,^{-1}$. In other words, if $\mathcal{R}$ is the set of all those subsets of $W\times W$ which meet these three conditions, then 
\[R=\bigcap_{T\in\mathcal{R}}T.\]
\end{definition}

The equivalence relation $R$ is well-defined since $W\times W\in\mathcal{R}$ and $\mathcal{R}$ is not empty.

\begin{definition}\label{canonical}
For any $w\in W$, we denote by $\overline{w}$ the $R$-equivalence class containing $w$. The map $\overline{\textcolor{white}{a}}:W\rightarrow W/R,w\mapsto \overline{w}$ is called \emph{the generalized canonical map} while its restriction to $S$, \emph{i.e.}, the map $\overline{\textcolor{white}{a}}:S\rightarrow W/R,s\mapsto \overline{s}$, is called \emph{the canonical map} . A subset $B\subset W$ is called $R$-\emph{saturated} if for any $w\in W$, $w\in B$ implies $\overline{w}\subset B$.
\end{definition}

The following proposition provides another definition of the semialgebraic-oriented equivalence relation.
\begin{proposition}\label{prop:anotherDEFofR}
Set $R_0=(\cup_{i=0}^6A_i)\bigcup\{(a, a)\,|\,a\in W\}$ and define the recurrence sequence $\{R_i\}_{i\in\mathbb{N}}$ as
\begin{equation}\label{Ri}
\left.
\begin{array}{rcl}
R_{i+1}&\hspace{-1mm}=&\hspace{-1mm}R_i\cup(R_i\oplus R_i)\cup(R_i\circledast R_i)\\
&&\hspace{-4.5mm}\cup\;(\mathbb{R}_+\odot R_i)\cup(R_i\backslash((\mathcal{O}\times W)\cup(W\times\mathcal{O})))^{-1}\\
&&\hspace{-4.5mm}\cup\;\text{rev}(R_i)\cup\{(a,c)\in W\times W\,|\,\exists b\in W\text{ \emph{s.t.} }(a,b),(b,c)\in R_i\},
\end{array}
\right.
\end{equation}
with 
\begin{equation}\label{equ:conponents}
\left\{
\begin{array}{lcl}
     R_i\oplus R_i&=&\{(a\oplus c, b\oplus d)\,|\,(a, b), (c, d)\in R_i\}, \\
          R_i\circledast R_i&=&\{(a\circledast c, b\circledast d)\,|\,(a, b), (c, d)\in R_i\},  \\
          \mathbb{R}_+\odot R_i&=&\{(r\odot a, r\odot b)\,|\,r\in\mathbb{R}_+, (a, b)\in R_i\},\\
          (R_i\backslash((\mathcal{O}\times W)\cup(W\times\mathcal{O})))^{-1}&=&\{(a^{-1}, b^{-1})\,|\,a, b\in W\backslash\mathcal{O}, (a, b)\in R_i\},\\
     \text{rev}(R_i)&=&\{(b,a)\,|\,(a,b)\in R_i\}.
\end{array}
\right.
\end{equation}
Then, $R=\cup^\infty_{i=0}R_i$ for $R$ defined in Definition \ref{def:R}.
\end{proposition}
\begin{proof}
Set $\hat{R}=\cup^\infty_{i=0}R_i$. By the definition of $R_i$, we see that $\hat{R}$ is an equivalence class on $W$ which contains the set $\cup_{i=0}^6A_i$ and is closed under the formal operations $\oplus$, $\circledast$, $\odot$ and $\;^{-1}$. Hence $R\subset\hat{R}$ by Definition \ref{def:R}. To prove the converse containment, it suffices to show $R_i\subset R$ for any $i\in\mathbb{N}$. This can be done inductively on the index $i$, which only contains routine checks; therefore, we omit it. The trick is to choose any $p\in R_k$ and then to proof $p\in R$. When $p$ is of one of the following forms \[(a\oplus c, b\oplus d), (a\circledast c, b\circledast d), (r\odot a, r\odot b), (a^{-1}, b^{-1}), (b,a)\] shown in equations (\ref{equ:conponents}) with $(a, b), (c, d)\in R_{k-1}$, then the induction step works for this case. Otherwise, $p=(a, c)$ with $(a, b), (b, c)\in R_{k-1}$ for some $b\in W$. In this case, the induction step also works.
\end{proof}
In the rest of the paper, we may use the definition of $R$ given in Proposition \ref{prop:anotherDEFofR} rather than the one in Definition \ref{def:R}, while conducting some proofs. For instance, we show in  Lemma \ref{Osat}  that the set $\mathcal{O}$ is $R$-saturated. And with that lemma, the following proposition is then straightforward. It is, however, important while define the semialgebra of fractions in the next subsection.

\begin{proposition}\label{0class}
It holds that $\overline{0}_S=\mathcal{O}$.
\end{proposition}
\begin{proof}
Since $\mathcal{O}\times\{0_S\}\subset R$, $\overline{0}_S\supset\mathcal{O}$. Conversely, set $x\in\overline{0}_S$, then $(x,0_S)\in R$. Clearly, $(x,0_S)\in R_i$ for some $i\in\mathbb{N}$ and $0_S\in\mathcal{O}$. By Lemma \ref{Osat}, we have $x\in\mathcal{O}$.
\end{proof}
\subsection{Defining the semialgebra of the fractions of $S$}\label{subsec:frac}
Set $F=W/R$ to be the set of equivalence classes in $W$, to turn it into a semialgebra, and we can define an addition, a multiplication, an inversion, and a scalar multiplication in $F$ as follows, 
\[
\left.
\begin{array}{rcl}
\overline{a}+\overline{b}&=&\overline{a\oplus b}\\
\overline{a}*\overline{b}&=&\overline{a\circledast b}\\
r\cdot \overline{a}&=&\overline{r\odot a}\\
\overline{c}^{-1}&=&\overline{c^{-1}}
\end{array}
\right\},
\;\forall a,b\in W,\;\forall r\in\mathbb{R}_+,\;\forall c\in W\backslash\mathcal{O}.
\]
\begin{proposition}\label{prop:fraction}
The operations given above are well-defined, and they turn $F$ into a semialgebra with any nonzero element invertible.
\end{proposition}
\begin{proof}
The definitions of addition, multiplication, and scalar multiplication are all well-defined since $R$ is closed under the formal operations $\oplus, \circledast$, and $\odot$. The relations in the set $A_1\cup A_2\cup A_3$ ensure that all semialgebra axioms are valid in $F$. Moreover, $\overline{0}_S$ and $\overline{1}_S$ are the identities of addition and multiplication, respectively.

We need to show that the inversion is well-defined and it is indeed ``the inversion" for the multiplication defined here.

Set $e\in W$ and suppose that $\overline{e}=\overline{c}$. By Proposition \ref{0class}, $c\notin\mathcal{O}$ implies $\overline{c}\neq\overline{0}_S$. Hence $\overline{e}\neq\overline{0}_S$, meaning that $e\notin\mathcal{O}$. Thus $e^{-1}\in W$. Moreover, $\overline{e^{-1}}=\overline{c^{-1}}$ since $R$ is closed under inversion. Therefore, the inversion is well-defined for any $c\notin\mathcal{O}$, \emph{i.e.}, for any $\overline{c}\neq \overline{0}_S$. Finally, the relations in the set $A_4$ guarantee  that the following equations hold for any $\overline{a},\overline{b}\neq\overline{0}_S$: $\overline{a}^{-1}*\overline{a}=\overline{1}_S=\overline{a}*\overline{a}^{-1}$, $(\overline{a}*\overline{b})^{-1}=\overline{b}^{-1}*\overline{a}^{-1}$, meaning that this is exactly the inversion for the multiplication ``$*$".
\end{proof}
Therefore, we have the following important definition:
\begin{definition}\label{def:fraction}
The \emph{semialgebra of the fractions} of a semialgebra $S$ satisfying Assumption \ref{assfree}, denoted by $F(S)$ $($or simply by $F)$, refers to the semialgebra $W(S)/R(S)$ proposed in Proposition \ref{prop:fraction}.
\end{definition}
The proposition below is clear since $A_5\subset R$:
\begin{proposition}
The canonical map from $S$ to $F$ is a semialgebra homomorphism.
\end{proposition}
\begin{remark}
The semialgebra $F$ is also zero-sum-free and zero-divisor-free: 
\begin{itemize}
    \item 
$\overline{a}+\overline{b}=\overline{0}_S$ implies that $a\oplus b\in\mathcal{O}$. Hence $a,b\in\mathcal{O}$, that is, $\overline{a}=\overline{b}=\overline{0}_S$; 
\item  $\overline{a}*\overline{b}=\overline{0}_S$ implies that $a\circledast b\in\mathcal{O}$. Hence $a\in\mathcal{O}$ or $b\in\mathcal{O}$, that is, $\overline{a}=\overline{0}_S$ or $\overline{b}=\overline{0}_S$;
\item  $r\cdot\overline{b}=\overline{0}_S$ implies that $r\odot b\in\mathcal{O}$. Hence $r=0_\mathbb{R}$ or $b\in\mathcal{O}$, that is, $r=0_\mathbb{R}$ or $\overline{b}=\overline{0}_S$.
\end{itemize}
\end{remark}
The following proposition is natural.
\begin{proposition}\label{prop:subsemialgebra}
Suppose that $S$ is a semialgebra satisfying Assumption \ref{assfree} and that $E\subset S$ is a sub-semialgebra with $0_S$ and $1_S$ its additive and multiplicative identities. Let $U(S), W(S), \mathcal{O}(S)$ and $R(S)$ be the sets defined in Definitions \ref{def:U}, \ref{def:WO} and \ref{def:R} for the semialgebra $S$, and $U(E), W(E), \mathcal{O}(E)$ and $R(E)$ the corresponding sets defined for the sub-semialgebra $E$. Then, we have 
\begin{equation}\label{equ:contained}
\left\{
\begin{array}{ccc}
     U(E)&\subset& U(S),  \\
     W(E)&\subset& W(S),  \\
     \mathcal{O}(E)&\subset& \mathcal{O}(S),\\
     W(E)\backslash\mathcal{O}(E)&\subset& W(S)\backslash\mathcal{O}(S),\\
     R(E)&\subset& R(S).
\end{array}
\right.
\end{equation}
Moreover, there is a natural semialgebra homomorphism between the following two semialgebra of fractions: $F(E)\rightarrow F(S): \overline{a}^E\mapsto\overline{a}^S$, where $a\in W(E)$ is any ``legal" formal expressions in the elements of $E$, and $\overline{\textcolor{white}{a}}^E: W(E)\rightarrow F(E)$ and $\overline{\textcolor{white}{a}}^S: W(S)\rightarrow F(S)$ are the generalized canonical maps.
\end{proposition}
\begin{proof}
By Definition \ref{def:U}, it is clear that $U(E)\subset U(S)$ since $E\subset S$. For any $x\in U(E)\subset U(S)$, we denote by $|x|$ the number of formal operations in $x$. Then, to prove the second, the third, and the fourth containment in inequalities (\ref{equ:contained}), it suffices to show that for any $i\in\mathbb{N}$ and any $x\in U(E)$ with $|x|=i$, we have
\begin{equation}\label{equ:implies}
\left\{
\begin{array}{rcl}
     x\in W(E) &\text{ implies }& x\in W(S),\\
          x\in \mathcal{O}(E) &\text{ implies }& x\in \mathcal{O}(S),\\
x\in W(E)\backslash\mathcal{O}(E)&\text{ implies }& x\in W(S)\backslash\mathcal{O}(S).
\end{array}
\right.
\end{equation}
That can be done inductively on the index $i$. We omit the details as usual. The trick is, $x$ is always of one of forms: $a\oplus b, a\circledast b, r\odot a$ and $a^{-1}$. In every case we have $|a|<|x|$ and $|b|<|x|$. Hence the inductive step works. One should also notice that implications in  (\ref{equ:implies}) sometimes ``prove each other". For instance, if $x=a^{-1}\in W(E)$ and we want to show $x\in W(E)$. Then we have $|a|<|x|$ and $a\in W(E)\backslash\mathcal{O}(E)$. By the inductive assumption, we have $a\in W(S)\backslash\mathcal{O}(S)$. Thus $x=a^{-1}\in W(S)$.

Suppose $A_1(S),\ldots,A_6(S)$ are the sets defined \emph{w.r.t.} the semialgebra $S$ in equations (\ref{semialgebraLaws})--(\ref{O=0}), and \[R_0(S)=(\cup_{i=0}^6A_i(S))\bigcup\{(a, a)\,|\,a\in W(S)\}\]is as in Proposition \ref{prop:anotherDEFofR}. Similarly, we have $A_1(E),\ldots,A_6(E)$ and \[R_0(E)=(\cup_{i=0}^6A_i(E))\bigcup\{(a, a)\,|\,a\in W(E)\}\] defined for the semialgebra $E$. From the first four inequalities in (\ref{equ:contained}), we see that $A_i(E)\subset A_i(S)$ for any $1\leq i\leq 6$. Thus $R_0(E)\subset R_0(S)$. Using the definition of semialgebra-oriented equivalence relation given in Proposition \ref{prop:anotherDEFofR}, we have the sequences $\{R_i(S)\}$ and $\{R_i(E)\}$ such that $R(S)=\cup_iR_i(S)$ and $R(E)=\cup_iR_i(E)$. To prove $R(E)\subset R(S)$, it suffices to show $R_i(E)\subset R_i(S)$ for any $i\in\mathbb{N}$, which can be done inductively on the index $i$. That is similar to proving $``R_i\subset R, \;\forall i\in\mathbb{N}"$ in Proposition \ref{prop:anotherDEFofR}. 

Now that we have $R(E)\subset R(S)$, the map $\overline{a}^E\mapsto\overline{a}^S$ is well-defined and is clearly a semialgebra homomorphism.
\end{proof}
\subsection{Related work}
The new definition of the semialgebra of fractions in this section generalizes the one in \cite[Proposition  11.5]{golan2013semirings}, where elements in the semiring of the fractions of a noncommutative semiring are only allowed to be of the form $a^{-1}*b$. While our definition allows the inversion to appear in any position or even to be nested in an element, \emph{e.g.}, we have elements of the form $a^{-1}*b*c^{-1}+d*p^{-1}*q$ and $(a+b^{-1})^{-1}$, \emph{etc.}, which are much more complicated.

In \cite[Chap.\,18]{golan2013semirings}, the author  also generalized the definition of the semiring of fractions in \cite[Chap.\,11]{golan2013semirings} via the concept of a \emph{Gabriel filter} of a semiring. Although our definition in this section generalizes the one in \cite[Chap.\,11]{golan2013semirings}, we do not know how our definition and the one given in \cite[Chap.\,18]{golan2013semirings} could possibly be related since they look utterly different from each other. The advantage of our definition is that it is elementary. Moreover, it is more convenient to prove the main theorem using our definition than the Gabriel filter.

Another topic related to the content in this section is the theory of \emph{skew fields} \cite{cohn1995skew}, which mainly studies the rings (or fields) of fractions of commutative and noncommutative rings. While many methods exist to construct skew fields of noncommutative rings, this section illustrates a particular way to construct a ``semi-skew field" of a noncommutative semialgebra. We do not use the theory of skew fields directly for two reasons: i) There are subtractions in those skew fields that are not allowed in semirings and semialgebras. One has to rule out the subtraction while defining  the semialgebra of fractions via the definition of the skew field, which is more troublesome than defining it from nothing. ii) One of the most important noncommutative skew fields is the fractions of noncommutative polynomials. Two elements in such a skew field are equal  if and only if  their values coincide with each other at every matrix tuple in a certain set. This definition of equality relies on other algebraic structures but not the skew field itself. That is, this kind of  equality  is not an ``intrinsic" property. The advantage of our definition of equality (\emph{i.e.}, the set $R$) relies only on the semialgebra $S$ itself, which is ``intrinsic".

In Proposition \ref{isom}, we  show that when the semialgebra $S$ is commutative, the semialgebra $F$ of its fractions defined in this section coincides with the corresponding concept in the commutative case, which can be obtained naturally from 
\cite[Example 11.7]{golan2013semirings}.  Therefore, 
the new definition of the semialgebra of fractions generalizes the corresponding concept in the commutative case, indicating that our definition is reasonable.

\section{Preordered Semialgebra and the Derived Preorder Relation on the Semialgebra of its Fractions}\label{sec:derivedPreorder}
This section discusses how a preorder relation on a semialgebra $S$ derives a relation on the semialgebra $F=F(S)$ of the fractions of $S$.
 
In a semialgebra, we are only concerned with preorders that are compatible with all the semialgebra operations (that is, which satisfy the implication in \ref{preoc}). Denote by $\mathcal{P}$ the set of such preorders on the semialgebra $F(S)$, then the intersection of any elements of $\mathcal{P}$ is still an element of $\mathcal{P}$. On the other hand, an inequality in $S$ is naturally expected to be true in $F$. Therefore,  we have the following definition:
\begin{definition}\label{def:derivedPreorder}
Let $(S, \leq)$ be a preordered semialgebra, and $F$ be the semialgebra of the fractions of $S$. Set $\mathscr{I}=\{(\overline{x}, \overline{y})\,|\,x\leq y\text{ in }S\}$. Then the \emph{derived preorder} ``$\preccurlyeq$" on $F$ \emph{w.r.t.} the preorder ``$\leq$" is the minimal preorder which is compatible with all the semialgebra operations and contains the set $\mathscr{I}$.  That is,
\begin{equation}\label{equ:derivedPreorder}
\{(a,b)\in F\times F\,|\,a\preccurlyeq b\} =\bigcap_{\mathscr{I}\subset P\in\mathcal{P}}P.
\end{equation}
\end{definition}
\begin{definition}
For any $a, b\in F$, if there are $k\in\mathbb{N}_{\geq1}$, $g_1, \ldots, g_k, h_1,\ldots,h_k\in F$ and $A_1, \ldots, A_k, B_1,\ldots,B_k\in S$ such that $A_i\leq B_i$ for all $i$, and that \begin{equation}\label{equ:lessdot}\left\{
\begin{array}{rcl}
     a&=&\sum_{i=1}^kg_i\overline{A}_ih_i  \\
    b & =&\sum_{i=1}^kg_i\overline{B}_ih_i
\end{array}
\right.
\end{equation}
in $F$, then we write ``$a\lessdot b$". 
\end{definition}
The proposition below gives another perspective into the derived preorder:
\begin{proposition}\label{prop:reduce2leq}
If $a, b\in F$, then $a\preccurlyeq b$  if and only if there are $c_1, \ldots, c_k\in F$ for some $k\in\mathbb{N}$ such that 
\begin{equation}\label{equ:manyLessdot}
a\lessdot c_1\lessdot\cdots\lessdot c_k\lessdot b.
\end{equation}
\end{proposition}
\begin{proof}
\noindent{}The ``if" part: This is clear from the observation that for any two elements $w_1, w_2$ in $F$, $w_1\lessdot w_2$ implies $w_1\preccurlyeq w_2$, which follows from the fact that the preorder $\preccurlyeq$ contains the set $\mathscr{I}$ and is compatible with all the semialgebra operations.\\\vspace{-3mm}

\noindent{}The ``only if" part: Define a binary relation ``$\preccurlyeq'$" on $F$ by the condition that for any $a, b\in F$, $a\preccurlyeq'b$ if and only if inequalities (\ref{equ:manyLessdot}) hold for some $k$ and some $c_i$. Since $k$ can be zero in this proposition, the binary relation $\lessdot$ is contained in the binary relation $\preccurlyeq'$. It suffices to show that the derived preorder is contained in the binary relation $\preccurlyeq'$. By the minimality of the derived preorder in Definition \ref{def:derivedPreorder}, it suffices to show that the binary relation $\preccurlyeq'$ is a preorder relation compatible with all the semialgebra operations and contains the set $\mathscr{I}$. The relation $\preccurlyeq'$ is reflexive since the relation $\lessdot$ is, while it is transitive by its definition. Thus it is a preorder relation. The set $\mathscr{I}$ is contained in the relation $\lessdot$, which is again contained in the preorder $\preccurlyeq'$. It is clear that the relation $\lessdot$ is compatible with all the semialgebra operations, from which one proves that this is also true for the preorder $\preccurlyeq'$.\end{proof}
The following corollary is interesting:
\begin{corollary}\label{coro:necessaryOfLeq}
If $a\preccurlyeq b\in F$, then there is a $w\in F$ such that $a+w\lessdot b+w$. That is, there are $k\in\mathbb{N}_{\geq1}$, $g_1, \ldots, g_k, h_1,\ldots,h_k\in F$ and $A_1, \ldots, A_k, B_1,\ldots,B_k\in S$ such that $A_i\leq B_i$ for all $i$, and  \[\left\{
\begin{array}{rcl}
     a+w&=&\sum_{i=1}^kg_i\overline{A}_ih_i,  \\
    b+w& =&\sum_{i=1}^kg_i\overline{B}_ih_i.
\end{array}
\right.
\]
\end{corollary}
\begin{proof}
One observes that $a\lessdot c\lessdot b$ implies $a+c\lessdot c+b$. Thus, there is some $w\in F$ (\emph{e.g.}, $w=c$) such that $a+w\lessdot b+w$. By reasoning inductively, one finds that condition (\ref{equ:manyLessdot}) also ensures the existence of such a $w$.
\end{proof}
 
The proposition below gives a recurrence characterization for the derived preorder.

\begin{proposition}
\label{prop:derivedPreorder}
For any $g,h\in F$, $g\preccurlyeq h$ if and only if the pair $(g,h)$ is in the set $L=\cup_{i=0}^\infty L_i$ with
\begin{equation}\label{Li}
\left\{
\begin{array}{lcl}
L_0&=&\{(\overline{x},\overline{y}),(\overline{0}_S,a),(a,a)\,|\,x,y\in S, x\leq y,a\in F\},\\
L_{i+1}&=&L_i\cup(\mathbb{R}_+\cdot L_i)\cup(L_i+L_i)\cup(F*L_i)\cup(L_i*F)\\
&&\cup\{(a,c)\,|\,\exists b\in F\text{ such that }(a,b),(b,c)\in L_i\},
\end{array}
\right.
\end{equation}
where \[\mathbb{R}_+\cdot L_i=\{(r\cdot b_1,r\cdot b_2)\,|\,(b_1,b_2)\in L_i,r\in\mathbb{R}_+\},\]\[L_i+L_i=\{(b_1+ c_1,b_2+c_2)\,|\,(b_1,b_2)\in L_i,(c_1,c_2)\in L_i\},\]\[F*L_i=\{(a* b_1,a* b_2)\,|\,(b_1,b_2)\in L_i,a\in F\},\] and\[L_i*F=\{(b_1* a,b_2* a)\,|\,(b_1,b_2)\in L_i,a\in F\}.\]
\end{proposition}
\begin{proof}
To prove that the derived preorder is contained in the binary relation $L$ via its minimality, we need to show that $L$ is a preorder relation that is compatible with all the semialgebra operations and contains the set $\mathscr{I}=\{(\overline{x},\overline{y})\,|\,x\leq y\text{ in }S\}$. It is reflexive since $(a,a)\in L_0$ for every $a\in F$, while its transitivity follows from the last component of the union in the definition of $L_{i+1}$ in equations (\ref{Li}). Therefore, $L$ is a preorder. Clearly, we have $\mathscr{I}\subset L_0\subset L$. The fact that $L$ is compatible with all the semialgebra operations is straightforward if one notices the components 
\[(\mathbb{R}_+\cdot L_i)\cup(L_i+L_i)\cup(F*L_i)\cup(L_i*F)\]in the definition of $L_{i+1}$.

Now we prove the converse, that is, every $(g,h)\in L$ satisfies $g\preccurlyeq h$. It suffices to show that for every natural number $i$, every $(g,h)\in L_i$ satisfies $g\preccurlyeq h$. This is clearly true once we prove it inductively on the index $i$. One thing that should be noticed while proving the $i=0$ case is that $\overline{0}_s\preccurlyeq a$ for every $a\in F$ since $\overline{0}_S\preccurlyeq\overline{1}_S$ and $a*\overline{0}_S\preccurlyeq a*\overline{1}_S$.
\end{proof}

In the rest of the paper, we sometimes use equations (\ref{Li}), \emph{i.e.}, the definition of the relation $L$, as the definition of the derived preorder.

Recall that, for a commutative zero-divisor-free semiring $K$ with preorder relation ``$\leq$", the  derived preorder relation ``$\preccurlyeq$" on the semiring $K^\text{fr}$ of the fractions of $K$ (see Appendix \ref{coincide}) is defined as follows: for any $x,y\in K$ and any $a,b\in K\backslash\{0_K\}$,
\begin{equation}\label{compare}
\frac{x}{a}\preccurlyeq\frac{y}{b}\text{ \hspace{2.5mm}if and only if \hspace{3mm} }\exists t\in K\backslash\{0_K\},x*b*t\leq y*a*t \text{ (\cite{fritz2021generalization}, page 13)}.
\end{equation}
If $K$ is a commutative zero-sum-free semialgebra, then the derived preorder relation  ``$\preccurlyeq$" can naturally be defined as in (\ref{compare}), too. In Appendix \ref{fracorder} Proposition \ref{preorderisom}, we  indicate that, when $S$ is commutative, and if one takes $K=S$ in (\ref{compare}), then the derived preorder relation for $F$ given in Definition \ref{def:derivedPreorder} coincides with the definition in (\ref{compare}) for $S^{\text{fr}}$. 
This explains the motivation of  Definition \ref{def:derivedPreorder}.

The following proposition is important for the proof of the main theorem:
\begin{proposition}\label{poweruniversal}
If $u\in S$ is power universal in $S$ $($with respect to $\leq)$, then $\overline{u}$ is power universal in $F$ $($with respect to $\preccurlyeq)$. 
\end{proposition}
\begin{proof}
We may assume that $u\neq0_S$, since otherwise, both the preorder relations $\leq$ on $S$ and $\preccurlyeq$ on $F$ would be trivial. In that case, the property that we want to prove is trivially true. Since $(S\backslash\{0_S\})\cap\mathcal{O}=\emptyset$, $u\neq0_S$ clearly implies $\lambda\cdot\overline{u}\neq\overline{0}_S$ for any fixed real number and $\lambda>1$. In the following, we first prove that $\lambda\cdot\overline{u}$ is power universal in $F$, then indicate that $\overline{u}$ is power universal, too.

For any $x\in W\backslash\mathcal{O}$ (or equivalently, any nonzero $\overline{x}\in F$), we need to show that there is a number  $k\in\mathbb{N}$ such that
\begin{equation}\label{lpu}
(\lambda\cdot \overline{u})^k*\overline{x}\succcurlyeq1_S,\;\;\overline{x}*(\lambda\cdot \overline{u})^k\succcurlyeq1_S\;\text{ and }\;(\lambda\cdot \overline{u})^k\succcurlyeq \overline{x}.
\end{equation}
Suppose that $x\in W_i$ for $i\in\mathbb{N}$ with $W_i$ defined as in Proposition \ref{prop:anotherDEFofWO}, we then prove inductively on the index $i$:
\\

\noindent{}\emph{The $``i=0"$ case:}

When $i=0$, $x\in W_0=S$ and $x\neq0_S$. Since $u$ is power universal in $S$, there is an integer $k\in\mathbb{N}$ such that inequalities (\ref{PU}) hold. By the definition of $\preccurlyeq$ in equations (\ref{Li}), we also have
\[
\overline{u}^k*\overline{x}\succcurlyeq1_S,\;\;\overline{x}* \overline{u}^k\succcurlyeq1_S\;\text{ and }\;\overline{u}^k\succcurlyeq \overline{x}.
\] Since $\lambda>1$, $\lambda\cdot\overline{u}=\overline{u}+(\lambda-1)\cdot\overline{u}\succcurlyeq\overline{u}$. Thus $(\lambda\cdot\overline{u})^k\succcurlyeq\overline{u}^k$, and we obtain the inequalities in $(\ref{lpu})$.
\\

\noindent{}\emph{The induction step:}

Assume that for any $0\leq i\leq \ell$ ($\ell\in\mathbb{N}$) and any $x\in W_i\backslash\mathcal{O}$, there is an integer  $k\in\mathbb{N}$ such that the inequalities in (\ref{lpu}) hold. We prove that this is also valid for $i=\ell+1$:

Set $x\in W_{\ell+1}\backslash\mathcal{O}$. Then there are several cases corresponding to the components in the definition of $W_{\ell+1}$ in equation (\ref{W}):

1) $x\in W_\ell$: The conclusion follows directly from the inductive assumption.

2) $x\in W_\ell\oplus W_\ell$: Then $\overline{x}=\overline{b}+\overline{c}$ with $b,c\in W_\ell$. Since $\overline{x}\neq\overline{0}_S$, $\overline{b}\neq\overline{0}_S$ or $\overline{c}\neq\overline{0}_S$. If one of $\overline{b}$ and $\overline{c}$ equals $\overline{0}_S$, say, $\overline{b}=\overline{0}_S$, then $\overline{x}=\overline{c}$. According to the inductive assumption, for $c\in W_\ell$, there is an integer $k\in\mathbb{N}$ such that (\ref{lpu}) hold with ``$\overline{x}$" replaced by ``$\overline{c}$".  
Therefore, inequalities in (\ref{lpu}) hold for $\overline{x}$ itself as well. Now, set $\overline{b}\neq\overline{0}_S$ and $\overline{c}\neq\overline{0}_S$, then we have $k_1,k_2\in\mathbb{N}$ such that 
\begin{equation}\label{six}
\left\{
\begin{array}{rcl}
(\lambda\cdot \overline{u})^{k_1}*\overline{b}\succcurlyeq\overline{1}_S,&\overline{b}*(\lambda\cdot \overline{u})^{k_1}\succcurlyeq\overline{1}_S,&(\lambda\cdot \overline{u})^{k_1}\succcurlyeq\overline{b},\\
(\lambda\cdot \overline{u})^{k_2}*\overline{c}\succcurlyeq\overline{1}_S,&\overline{c}*(\lambda\cdot \overline{u})^{k_2}\succcurlyeq\overline{1}_S,&(\lambda\cdot \overline{u})^{k_2}\succcurlyeq\overline{c}.
\end{array}
\right.
\end{equation}
Set $\hat{k}=\max\{k_1,k_2\}$, and choose sufficiently large  $k\geq \hat{k}$ satisfying $\lambda^k\geq2\lambda^{\hat{k}}$, then
\[
(\lambda\cdot \overline{u})^k*\overline{x}\succcurlyeq(\lambda\cdot \overline{u})^{\hat{k}}*\overline{x}=(\lambda\cdot \overline{u})^{\hat{k}}*(\overline{b}+\overline{c})\succcurlyeq\overline{1}_S+\overline{1}_S\succcurlyeq\overline{1}_S.
\]
Similarly, we have $\overline{x}*(\lambda\cdot \overline{u})^k\succcurlyeq\overline{1}_S$. On the other hand,
\[
(\lambda\cdot \overline{u})^k=\lambda^k\cdot\overline{u}^k\succcurlyeq(2\lambda^{\hat{k}})\cdot\overline{u}^{\hat{k}}=2\cdot(\lambda\cdot\overline{u})^{\hat{k}}\succcurlyeq\overline{b}+\overline{c}=\overline{x}.
\]

3) $x\in W_\ell\circledast W_\ell$: Then $\overline{x}=\overline{b}*\overline{c}$ with $b,c\in W_\ell\backslash\mathcal{O}$ and the inequalities (\ref{six}) hold for some $k_1,k_2\in\mathbb{N}$. Set $k=k_1+k_2$. Multiplying the rightmost inequalities in each rows of (\ref{six}), we get $(\lambda\cdot \overline{u})^{k}\succcurlyeq\overline{b}*\overline{c}=\overline{x}$. Multiplying the leftmost inequality in the first row and the middle inequality in the second row of (\ref{six}), we obtain $(\lambda\cdot \overline{u})^{k_1}*\overline{b}*\overline{c}*(\lambda\cdot \overline{u})^{k_2}\succcurlyeq\overline{1}_S$. Then, by multiplying $(\lambda\cdot \overline{u})^{-k_1}$ and $(\lambda\cdot \overline{u})^{-k_2}$ from the left and the right sides respectively, one derives $\overline{x}\succcurlyeq(\lambda\cdot \overline{u})^{-k}$. From that, one obtains both $(\lambda\cdot \overline{u})^{k}*\overline{x}\succcurlyeq\overline{1}_S$ and $\overline{x}*(\lambda\cdot \overline{u})^{k}\succcurlyeq\overline{1}_S$ by simply multiplying $(\lambda\cdot \overline{u})^{k}$ from the left and the right sides respectively.

4) $x\in\mathbb{R}_+\odot W_\ell$: In this case, $x=r\cdot\overline{b}$ with $b\in W_\ell\backslash\mathcal{O}$ and $r$ is a positive real number. Hence there is an integer $k_1\in\mathbb{N}$ such that the inequalities in the first row or (\ref{six}) hold. Choose a  sufficiently large integer $m\in\mathbb{N}$ satisfying $\lambda^m\geq\max\{r,1/r\}$ and set $k=m+k_1$. Then we have
\[
\left.
\begin{array}{rcl}
(\lambda\cdot\overline{u})^k*\overline{x}&\succcurlyeq&(\lambda\cdot\overline{1}_S)^m*(\lambda\cdot\overline{u})^{k_1}*(r\cdot\overline{b})\\
&\succcurlyeq&(\lambda^mr)\cdot((\lambda\cdot\overline{u})^{k_1}*\overline{b})\\
&\succcurlyeq&1_\mathbb{R}\cdot\overline{1}_S\\
&=&\overline{1}_S,
\end{array}
\right.
\]
and similarly
\[
\left.
\begin{array}{rcl}
\overline{x}*(\lambda\cdot\overline{u})^k&\succcurlyeq&(r\cdot\overline{b})*(\lambda\cdot\overline{u})^{k_1}*(\lambda\cdot\overline{1}_S)^m\\
&\succcurlyeq&(\lambda^mr)\cdot(\overline{b}*(\lambda\cdot\overline{u})^{k_1})\\
&\succcurlyeq&1_\mathbb{R}\cdot\overline{1}_S\\
&=&\overline{1}_S.
\end{array}
\right.
\]
Moreover, we have $(\lambda\cdot\overline{u})^k=(\lambda\cdot\overline{u})^m*(\lambda\cdot\overline{u})^{k_1}\succcurlyeq(\lambda\cdot\overline{1}_S)^m*\overline{b}=\lambda^m\cdot\overline{b}\succcurlyeq r\cdot\overline{b}=\overline{x}$.

5) $x\in(W_\ell\backslash\mathcal{O}_\ell)^{-1}$: Then $\overline{x}=\overline{b}^{-1}$ with $b\in W_\ell\backslash\mathcal{O}_\ell$ satisfying the inequalities in the first row of (\ref{six}) for some integer $k_1\in\mathbb{N}$. From the rightmost inequality we obtain both $(\lambda\cdot\overline{u})^{k_1}*\overline{b}^{-1}\succcurlyeq\overline{1}_S$ and $\,\overline{b}^{-1}*(\lambda\cdot\overline{u})^{k_1}\succcurlyeq\overline{1}_S$ by multiplying $\overline{b}^{-1}$ from its right and left sides, respectively. From the leftmost one we obtain $(\lambda\cdot\overline{u})^{k_1}\succcurlyeq\overline{b}^{-1}$ by multiplying $\overline{b}^{-1}$ from its right side.
\\\vspace{-2mm}

By now we have shown that $\lambda\cdot\overline{u}$ is power universal in $F$ since the inequalities in (\ref{lpu}) hold. Now we show that $\overline{u}$ is also power universal. Apparently, $(\lambda\cdot1_S)^k\neq0_S$, hence there is $k_0\in\mathbb{N}$ such that $u^{k_0}\geq(\lambda\cdot1_S)^k$. Therefore $\overline{u}^{k_0}\succcurlyeq(\lambda\cdot\overline{1}_S)^k$ in $F$. Using the inequalities in (\ref{lpu}), we have $\overline{u}^{k_0+k}\succcurlyeq(\lambda\cdot\overline{1}_S)^k*(\lambda^{-k}\cdot\overline{x})=\overline{x}$, and $\overline{u}^{k_0+k}*\overline{x}\succcurlyeq(\lambda\cdot\overline{1}_S)^k*\overline{u}^k*\overline{x}=(\lambda\cdot\overline{u})^k*\overline{x}\succcurlyeq1_S$. Similarly, we also have $\overline{x}*\overline{u}^{k_0+k}\succcurlyeq1_S$. Hence $\overline{u}$ is indeed power universal.
\end{proof}
\section{Extending $\mathbb{R}_+$-Valued Homomorphisms from a Semialgebra to the Semialgebra of its Fractions}\label{sec:extend}
This section explains how a monotone homomorphism from the preordered semialgebra $S$ to some other preordered semialgebra $T$ (which shares some good properties) can be extended to the semialgebra $F$ of the fractions of $S$. In particular, the extension is always guaranteed when $T=\mathbb{R}_+$.

The following lemma due to \cite{fritz2021generalization} is straightforward:
\begin{lemma}\label{trivialnull}
Suppose that $(T,\leq_{_T})$ is a preordered semialgebra with nontrivial preorder relation $\leq_{_T}$ such that $1_T\geq_{_T}\hspace{-0.7mm}0_T$, and  $(E,\leq_{_E})$ is a preordered semialgebra with $1_E\geq_{_E}0_E$ and a power universal element $u$. Then for any monotone semialgebra homomorphism $f: E\rightarrow T$ and any nonzero $x\in E$, we have $f(x)\neq0_T$.
\end{lemma}
\begin{proof}
There is an integer $k\in\mathbb{N}$ satisfying  
\begin{equation}\label{ukk}
u^k*x\geq_{_E}1_E.
\end{equation}
Applying the homomorphism $f$ to the inequality (\ref{ukk}), we get $f(u)^k*f(x)\geq_{_T}\hspace{-0.7mm}1_T$. Assume that $f(x)=0_T$, then we have $0_T\geq_{_T}\hspace{-0.7mm}1_T$. For any $y,z\in T$, we obtain
\begin{equation}\label{tri}
y=y*1_T\leq_{_T}\hspace{-0.7mm}y*0_T=z*0_T\leq_{_T}\hspace{-0.7mm}z*1_T=z,\end{equation}
which contradicts the assumption that $\leq_{_T}$ is nontrivial.
\end{proof}

Let $T$ be given  as in Lemma \ref{trivialnull} and suppose that $(S,\leq)$ is a preordered semialgebra satisfying Assumption \ref{power}. We assume further, throughout the rest of the paper, that $T$ is both zero-sum-free and zero-divisor-free, and that every nonzero element of $T$ is invertible (\emph{e.g.}, if $T=\mathbb{R}_+$, then all the assumptions on the semialgebra $T$ are satisfied).

For any monotone semialgebra homomorphism $f: S\rightarrow T$, the following proposition defines a map $f^W$ from the set of legal formal rational expressions $W(S)$ (Definition \ref{def:WO}) to the semialgebra $T$: 
\begin{proposition}\label{fW}
Suppose $f: S\rightarrow T$ is a monotone semialgebra homomorphism with preordered semialgebras $S$ and $T$ described as above. Let $\{W_i\}$ be the sequence defined in Proposition \ref{prop:anotherDEFofWO} such that $W=\cup_iW_i$. Then, the recursive assignment below for $f^W$ makes it a well-defined map from $W$ to $T$:

1) For any $x\in W_0=S$, set $f^W(x)=f(x)$;

2) Otherwise, suppose that $x\in W_\ell$ with $\ell=\min\{i\in\mathbb{N}\,|\,x\in W_i\}>0$, then set
\begin{equation}\label{assgn}
f^W(x)=\left\{
\begin{array}{rcl}
f^W(a)+f^W(b),&if&x=a\oplus b\in W_{\ell-1}\oplus W_{\ell-1},\\f^W(a)*f^W(b),&if&x=a\circledast b\in W_{\ell-1}\circledast W_{\ell-1},\\
r\cdot f^W(a),&if&x=r\odot a\in \mathbb{R}_+\odot W_{\ell-1},\\
(f^W(a))^{-1},&if&x=a^{-1}\in (W_{\ell-1}\backslash\mathcal{O}_{\ell-1})^{-1}.
\end{array}
\right.
\end{equation}
\end{proposition}
\begin{proof}
 It is sufficient to prove (inductively on the index $i$) that for any $i\in\mathbb{N}$ and any $x\in W_{i}$, $f^W(x)$ is a well-defined element in $T$ and $f^W(x)\neq0_T$ whenever $x\notin\mathcal{O}_i$ (thus the assignment in the last row of (\ref{assgn}) is well-defined due to the assumption that every nonzero element in $T$ is invertible and Lemma \ref{trivialnull}):

If $i=0$, $x\in W_0=S$, and $f^W(x)=f(x)$ is well-defined. When $x\notin\mathcal{O}_0=\{0_S\}$, we have $x\in S\backslash\{0_S\}$. Then Lemma \ref{trivialnull} indicates that $f^W(x)=f(x)\neq0_T$.

Suppose that for any $0\leq i<k$ ($k\in\mathbb{N}_{\geq 1}$) and any $x\in W_i$, $f^W(x)$ is a well-defined element in $T$ and $f^W(x)\neq0_T$ whenever $x\notin\mathcal{O}_i$. We then show that this is also valid for $i=k$ and any $x\in W_i$:

Let  $x\in W_k$, we may assume that $x\notin W_{k-1}$ (Since, if $x\in W_{k-1}$, then the inductive assumption implies the conclusions that we need). Hence $k\geq1$ is the minimal nature number such 
that $x\in W_k$. 
According to equation (\ref{W}), there are several cases we need to check:

i) If $x=a\oplus b\in W_{k-1}\oplus W_{k-1}$, then $f^W(x)=f^W(a)+f^W(b)$ is well-defined since $f^W(a)$ and $f^W(b)$ are well-defined. When $x\notin\mathcal{O}_k$, we have either $a\notin\mathcal{O}_{k-1}$ or $b\notin\mathcal{O}_{k-1}$. Hence one of $f^W(a)$ and $f^W(b)$ is nonzero in $T$. We conclude that $f^W(x)=f^W(a)+f^W(b)$ is also nonzero due to the fact that $T$ is zero-sum-free.
\vspace{0.5mm}

ii) If $x=a\circledast b\in W_{k-1}\circledast W_{k-1}$, then $f^W(x)=f^W(a)*f^W(b)$ is well-defined. When $x\notin\mathcal{O}_k$, we have $a,b\notin\mathcal{O}_{k-1}$. By the inductive assumption, both $f^W(a)$ and $f^W(b)$ are nonzero. Hence $f^W(x)=f^W(a)*f^W(b)$ is also nonzero (since $T$ is zero-divisor-free).

\vspace{0.5mm}

iii) If $x=r\odot a\in \mathbb{R}_+\odot W_{k-1}$, then $f^W(x)=r\cdot f^W(a)$ is well-defined. When $x\notin\mathcal{O}_k$, we have $r\neq0$ and $a\notin\mathcal{O}_{k-1}$. By the inductive assumption, $f^W(a)$ is nonzero in $T$. Hence $f^W(x)=r\cdot f^W(a)$ is also nonzero (using again the fact that $T$ is zero-divisor-free).
\vspace{0.5mm}

iv) If $x=a^{-1}\in (W_{k-1}\backslash\mathcal{O}_{k-1})^{-1}$ with $a\in W_{k-1}\backslash\mathcal{O}_{k-1}$, then $f^W(a)$ is well-defined and nonzero in $T$ by the inductive assumption. Since every nonzero element in $T$ is invertible, $f^W(x)=(f^W(a))^{-1}$ is also well-defined. Moreover, it is nonzero since $(f^W(a))^{-1}=0_T$ would imply \[1_T=f^W(a)*(f^W(a))^{-1}=f^W(a)*0_T=0_T,\] which indicates that $1_T\leq_{_T}0_T\leq_{_T}1_T$. And this would result in the triviality of $\leq_{_T}$ as shown in (\ref{tri}), contradicting one of the assumptions in Lemma \ref{trivialnull}.
\end{proof}
 Note that by formulae (\ref{assgn}) we have
\begin{equation}\label{prehomo}
\left\{
\begin{array}{rcl}
     f^W(a\oplus b)&=&f^W(a)+f^W(b), \\
     f^W(a\circledast b)&=&f^W(a)*f^W(b),\\
     f^W(r\odot a) &=&r\cdot f^W(a),\\
    f^W(c^{-1}) &=&(f^W(c))^{-1},
\end{array}
\right.
\end{equation}
for all $a,b\in W,c\in W\backslash\mathcal{O}$ and $r\in\mathbb{R}_+$, no matter what concrete value the letter $\ell$ in formulae (\ref{assgn}) takes. Therefore, Proposition \ref{fW} can be rewritten concisely as follows:
\begin{proposition}\label{fWAgain}
Suppose $f: S\rightarrow T$ is a monotone semialgebra homomorphism with preordered semialgebras $S$ and $T$ described as before. Then, the recursive assignment below for $f^W$ makes it a well-defined map from $W$ to $T$:

1) For any $x\in S$, set $f^W(x)=f(x)$;

2) If $x\in W\backslash S$ contains some formal operations, then set
\begin{equation}\label{assgnAgain}
f^W(x)=\left\{
\begin{array}{rcl}
f^W(a)+f^W(b),&if&x=a\oplus b\text{ with } a, b\in W,\\f^W(a)*f^W(b),&if&x=a\circledast b\text{ with } a, b\in W,\\
r\cdot f^W(a),&if&x=r\odot a\text{ with } r\in \mathbb{R}_+, a\in W,\\
(f^W(a))^{-1},&if&x=a^{-1}\text{ with } a\in W\backslash\mathcal{O}.
\end{array}
\right.
\end{equation}
\end{proposition}
Equations (\ref{prehomo}) are useful in the proof of the following proposition:

\begin{proposition}\label{f^F}
For any monotone semialgebra homomorphism $f: S\rightarrow T$, the map $f^F: F\rightarrow T,\, \overline{x}\mapsto f^W(x)\;(\forall x\in W)$ is well-defined.
\end{proposition}
\begin{proof}
For any $x,y\in W$ with $\overline{x}=\overline{y}$ (that is, $(x,y)\in R$ with $R$ in Definition \ref{def:R}), we will show that $f^W(x)=f^W(y)$. By Proposition \ref{prop:anotherDEFofR}, we have $R=\cup_{i=0}^\infty R_i$ with $R_i$ defined therein. Hence it is sufficient to prove that for any $i\in\mathbb{N}$ and any $x,y\in W$ satisfying $(x,y)\in R_i$, we have $f^W(x)=f^W(y)$. The proof is inductive on the index $i$:

\noindent{}\emph{The $``i=0"$ case:}

Note that $R_0=(\cup_{j=1}^6A_j)\bigcup\{(a, a)\,|\,a\in W\}$ as in Proposition \ref{prop:anotherDEFofR}, there are several sub-cases:
\begin{enumerate}
\item 
If $(x,y)\in \cup_{j=1}^4A_j\subset R_0$, $f^W(x)=f^W(y)$ follows from (\ref{prehomo}). 

\item 
If $(x,y)\in A_5$, then $x,y\in Q$ and $x_S=y_S$. We claim that, for any $z\in Q$, it holds that $f^W(z)=f(z_S)$. Instead of giving a formal proof, we only provide a simple example to illustrate the claim above: if $z=s_1\oplus((r\odot s_2)\circledast s_3)\in Q$ for some elements $s_1,s_2,s_3\in S$ and a real number  $r\in\mathbb{R}_+$, then clearly we have

\[
\left.
\begin{array}{rcl}
f^W(z)&=&f(s_1)+((r\cdot f(s_2))*f(s_3))\\
&=&f(s_1+((r\cdot s_2)*s_3))\\
&=&f(z_S),
\end{array}
\right.
\]
where the first equality is due to the definition (\ref{assgn}) of $f^W$ and the identities in (\ref{prehomo}), while the second one is valid since $f$ is a semialgebra homomorphism from $S$ to $T$. Thus $f^W(x)=f(x_S)=f(y_S)=f^W(y)$.

\item 
If $(x,y)\in \{(a,a)\,|\,a\in W\}$, then we have  $x=y$ and  $f^W(x)=f^W(y)$. 

\item 
If $(x,y)\in A_6$, then we have $x\in\mathcal{O}$ and $y=0_S$. Since $f^W(0_S)=0_T$, we need to show that $f^W(x)=0_T$ for any $x\in\mathcal{O}$: Clearly, $f^W(x)=0_T$ for any $x\in\mathcal{O}_0=\{0_S\}$. Assume that for any $0\leq \hat{i}<k$ ($k\in\mathbb{N}_{\geq1}$) and any $x\in\mathcal{O}_{\hat{i}}$, we have $f^W(x)=0_T$. We then prove that this is also valid for $\hat{i}=k$. Set $x\in\mathcal{O}_k$ and $x\notin\mathcal{O}_{k-1}$. Then, by the equation (\ref{O}), we have $x=a\oplus b\in\mathcal{O}_{k-1}\oplus\mathcal{O}_{k-1}$, $x=a\circledast b\in(\mathcal{O}_{k-1}\circledast W_{k-1})\cup(W_{k-1}\circledast\mathcal{O}_{k-1})$ or $x=r\odot a\in(\mathbb{R}_+\odot\mathcal{O}_{k-1})\cup(\{0_\mathbb{R}\}\odot W_{k-1})$. From the equations in (\ref{prehomo}), we see that $f^W(x)=0_T$ holds in all cases. 
\end{enumerate}
Hence $f^W(x)=f^W(y)$ whenever $(x,y)\in R_0$.

\vspace{1mm}

\noindent{}\emph{The inductive case:}

Suppose that for any $0\leq i<\ell$ ($\ell\in\mathbb{N}_{\geq1}$) and any $(x,y)\in R_i$, we have $f^W(x)=f^W(y)$. We then prove that this is also valid for $i=\ell$.

Assume without loss of generality that $(x,y)\in R_\ell\backslash R_{\ell-1}$. There are several cases left:
\begin{enumerate}

\item  If $(x,y)=(x_1,y_1)\oplus(x_2,y_2)\in R_{\ell-1}\oplus R_{\ell-1}$, then $x=x_1\oplus x_2$, $y=y_1\oplus y_2$ and $(x_1,y_1),(x_2,y_2)\in R_{\ell-1}$. Thus $f^W(x)=f^W(x_1)+f^W(x_2)=f^W(y_1)+f^W(y_2)=f^W(y)$.

\item  Similarly, we can deal with the cases where $x\in(R_{\ell-1}\circledast R_{\ell-1})\cup(\mathbb{R}_+\odot R_{\ell-1})$. 

\item  If $(x,y)=(x_1^{-1},y_1^{-1})\in(R_{\ell-1}\backslash((\mathcal{O}\times W)\cup(W\times \mathcal{O})))^{-1}$, then $x_1,y_1\notin\mathcal{O},x=x_1^{-1},y=y_1^{-1}$ and $(x_1,y_1)\in R_{\ell-1}$. Since $x_1,y_1\notin\mathcal{O}$, neither of $f^W(x_1)$ and $f^W(y_1)$ is zero according to the proof of Proposition \ref{fW}. Thus both $(f^W(x_1))^{-1}$ and $(f^W(y_1))^{-1}$ are well-defined and they are equal to each other since $f^W(x_1)=f^W(y_1)$ by the inductive assumption.

\item  If $(x,y)\in\text{rev}(R)_{\ell-1}$, then $(y,x)\in R_{\ell-1}$. The inductive assumption indicates that $f^W(y)=f^W(x)$.

\item  If there is $z\in W$ such that $(x,z),(z,y)\in R_{\ell-1}$, then $f^W(x)=f^W(z)$ and $f^W(z)=f^W(y)$ again by the inductive assumption.
\end{enumerate}
The above discussion shows that $f^F$ is a well-defined map from $F$ to $T$.
\end{proof}

\begin{proposition}\label{f^F2}
The map $f^F: F\rightarrow T$ is a monotone semialgebra homomorphism.
\end{proposition}
\begin{proof}
The fact that it is a semialgebra homomorphism  directly follows from the identities in (\ref{prehomo}): \emph{e.g.}, \[f^F(\overline{x}+\overline{y})=f^F(\overline{x\oplus y})=f^W(x\oplus y)=f^W(x)+f^W(y)=f^F(\overline{x})+f^F(\overline{y}).\] The rest three identities corresponding to the last three rows in (\ref{prehomo}) can be derived similarly. 

We now prove that $f^F$ is monotone.

Remember that the preorder relation $\preccurlyeq$ on $F$ is given by the set $L\subset F\times F$ defined in Proposition \ref{prop:derivedPreorder}. We need to show that $f^F(g)\leq_{_{T}}\hspace{-1mm}f^F(h)$ for any $(g,h)\in L=\cup_iL_i$ with $L_i$ defined in equations (\ref{Li}).

If $(g,h)\in L_0$, then $g=\overline{0}_S$, $g=h$ or $g=\overline{x},h=\overline{y}$ for some $x,y\in S$ such that $x\leq y$. The first two sub-cases are trivial while the last one follows from \[f^F(\overline{x})=f^W(x)=f(x)\leq_{_{T}}\hspace{-1mm}f(y)=f^W(y)=f^F(\overline{y}).\]

Assume that for any $0\leq i<k$ ($k\in\mathbb{N}_{\geq1}$) and any $(g,h)\in L_{i}$, we have $f^F(g)\leq_{_{T}}\hspace{-1mm}f^F(h)$. We then prove the case where $i=k$.

We may assume that $(g,h)\notin L_{k-1}$, thus
\begin{equation}\label{cases}
\left.
\begin{array}{rcl}
(g,h)&\in &L_{k}\backslash L_{k-1}\\
&=&\Big((\mathbb{R}_+\cdot L_{k-1})\cup(L_{k-1}+L_{k-1})\cup(F*L_{k-1})\cup(L_{k-1}*F)\\
&&\cup\{(g,h)\,|\,\exists d\in F\text{ so that }(g,d),(d,h)\in L_{k-1}\}\Big)\Big\backslash L_{k-1}.
\end{array}
\right.
\end{equation}
Those cases corresponding to the components in the formula (\ref{cases}) can be dealt with separately, by using the transitivity and the compatible properties (\ref{preoc}) of a preorder relation on a semialgebra. To illustrate, consider the case where $(g,h)=(g_1,h_1)+(g_2,h_2)\in L_{k-1}+L_{k-1}$. We have $g=g_1+g_2$, $h=h_1+h_2$ and
\[\left.\begin{array}{rll}
f^F(g)=f^F(g_1+g_2)&=&\hspace{-2mm}f^F(g_1)+f^F(g_2)\\
&\leq_{_T}&\hspace{-2mm}f^F(h_1)+f^F(g_2)\\
&\leq_{_T}&\hspace{-2mm}f^F(h_1)+f^F(h_2)=f^F(h_1+h_2)=f^F(h).
\end{array}\right.\]
The two $\leq_{_T}\hspace{-0.4mm}$'s above are due to the inductive assumption and the first rule in (\ref{preoc}). The other cases can be handled similarly, hence we omit the details.
\end{proof}
\begin{remark}
For any monotone semialgebra homomorphism $\eta: F\rightarrow T$ such that $\eta(\overline{x})=f^F(\overline{x})$ for all $x\in S$ $($\emph{i.e.}, it extends the monotone semialgebra homomorphism $f)$, we can prove inductively that $\eta=f^F$. This indicates the universal property of the semialgebra $F$ of the fractions of $S$.
\end{remark}
\begin{remark}\label{realvaluedextension}
Set $T=\mathbb{R}_+$, then Propositions \ref{f^F} and \ref{f^F2} indicate that any monotone semialgebra homomorphism $f: S\rightarrow \mathbb{R}_+$ can be extended to a monotone semialgebra homomorphism mapping from $F$ to $\mathbb{R}_+$.
\end{remark}
\section{From A Semialgebra to A Linear Space}\label{sec:V}
In this section, we define a preordered $\mathbb{R}$-linear space based on the semialgebra $F$. This definition is already given by \cite{fritz2021generalization}, while we rewrite it in detail for the completeness of the present paper.

Define $V=F-F$ to be the Grothendieck group \cite[Chap.\,II]{karoubi1978k} of the commutative monoid $(F,+)$, with each element in $V$ of the form $a-b$ for some $a,b\in F$, and with $a-b=a'-b'$ in $V$ if and only if $\exists c\in F$ \emph{s.t.} $a+b'+c=a'+b+c$ in $F$. Then, $V$ becomes an $\mathbb{R}$-linear space once we define a scalar multiplication on it as below:
\[
r\cdot(a-b)=\left\{
\begin{array}{rclcl}
r\cdot a\hspace{-4.2mm}&\hspace{-1mm}-&\hspace{-1.4mm}r\cdot b,&\text{if\hspace{-1.7mm}}&r\in\mathbb{R}_{\geq0},\\
(-r)\cdot b\hspace{-4.2mm}&\hspace{-1mm}-&\hspace{-1.4mm}(-r)\cdot a,&\text{if\hspace{-1.7mm}}&r\in\mathbb{R}_{<0}.
\end{array}
\right.
\]
It's straightforward to verify that this operation is well-defined. One also observes the following identities:
\begin{equation}\label{idn}
\left.
\begin{array}{rcl}
1_\mathbb{R}\cdot(a-b)&=&a-b,\\
(rs)\cdot(a-b)&=&r\cdot(s\cdot(a-b)),\\
r\cdot((a-b)+(c-d))&=&r\cdot(a-b)+r\cdot(c-d),\\
(r+s)\cdot(a-b)&=&r\cdot(a-b)+s\cdot(a-b)
\end{array}
\right.
\end{equation}
with $r,s\in\mathbb{R}$ and $a,b,c,d\in F$. All of the above identities can be verified straightforwardly. While most of them can be checked concisely, it is not the case for the last one: One may need to consider six cases corresponding to the different sign(s) of the real numbers $r,s$ and $r+s$. To illustrate, we deal with the case where $r\geq0,s<0$ and $r+s<0$ but leave the rest of them to the readers.

Now that $r+s<0$, 
\[(r+s)\cdot(a-b)=(-r-s)\cdot b-(-r-s)\cdot a\]
by the definition. On the other hand, we have
\[
\left.
\begin{array}{rcl}\vspace{1.6mm}
&&\textcolor{white}{(}r\cdot(a-b)+s\cdot(a-b)\\\vspace{1.6mm}
&=&(r\cdot a-r\cdot b)+\big((-s)\cdot b-(-s)\cdot a\big)\\
&=&\big(r\cdot a+(-s)\cdot b\big)-\big(r\cdot b+(-s)\cdot a\big).
\end{array}
\right.
\]
Then it is sufficient to show that
\[
\left.
\begin{array}{llclll}
    &(-r-s)\cdot b-(-r-s)\cdot a&=&\big(r\cdot a+(-s)\cdot b\big)-\big(r\cdot b+(-s)\cdot a\big)&\text{(in $V$)}  \\\vspace{1.6mm}
    \Leftarrow&(-r-s)\cdot b+r\cdot b+(-s)\cdot a&=&(-r-s)\cdot a+\big(r\cdot a+(-s)\cdot b\big)&\text{(in $F$)}\\\vspace{1.6mm}
    \Leftarrow&(-s)\cdot b+(-s)\cdot a&=&(-s)\cdot b+(-s)\cdot a.&\text{(in $F$)}
\end{array}
\right.
\]The last equality is trivially true. Hence the last identity in (\ref{idn}) is verified for the case of $r\geq0,s<0$ and $r+s<0$. 

The identities in (\ref{idn}) suggest that $V$ is indeed an $\mathbb{R}$-linear space. In the sequel we define a preorder relation on $V$ based on the preorder relation $\preccurlyeq$ on $F$.

For any $a,b,c,d\in F$ (or equivalently, for any $a-b,c-d\in V$), we write $a-b\unlhd c-d$ if and only if  there is a $g\in F$ so that $a+d+g\preccurlyeq b+c+g$ in $F$. It is straightforward to prove that the relation $\unlhd$ is well-defined and it is also a preorder. Furthermore, it is not hard to show that 
\begin{equation}\label{Ccone}
a-b\unlhd c-d\text{\hspace{1.6mm} implies }
\left\{
\begin{array}{rcl}
     (a-b)+(g-h)&\unlhd&(c+d)+(g-h), \\
     r\cdot(a-b)&\unlhd& r\cdot(c-d),
\end{array}
\right.
\end{equation}
for any $a,b,c,d,g,h\in F$ and any $r\in\mathbb{R}_+$. These inequalities are similar to the ones given in (\ref{preoc}).

To conclude, the $\mathbb{R}$-linear space $V$ is preordered with the relation $\unlhd$ which is compatible with the addition and the scalar multiplication defined in itself.

\section{The Main Theorem and the Proof}\label{sec:mainTHM}
This section is devoted to the proof of the main theorem. We state it again for convenience: 

\noindent\textbf{Theorem \ref{main}}\emph{
Let $(S,\leq)$ be a preordered semialgebra satisfying Assumptions \ref{power} and \ref{assfree}, with a power universal element $u\in S$. Suppose that $(F,\preccurlyeq)$ is the preordered semialgebra of the fractions of $S$ defined in Subsection \ref{subsec:frac} with ``$\preccurlyeq$" in Def. \ref{def:derivedPreorder} and  $\overline{{}\textcolor{white}{a}}: S\rightarrow F, x\mapsto \overline{x}$ the canonical map in Def. \ref{canonical}. Then, for every nonzero $x,y\in S$, the following are equivalent:}

\noindent\emph{
$($a$)$ $f(x)\geq f(y)$ for every monotone semialgebra homomorphism $f:S\rightarrow\mathbb{R}_+$.}

\noindent\emph{
$($b$)$ For every real number $\epsilon>0$, there is $m\in\mathbb{N}$ such that 
\begin{equation}\label{binq}
\overline{x}+\sum_{j=0}^m\epsilon^{j+1}\cdot\overline{u}^j\succcurlyeq\overline{y}.
\end{equation}}

\noindent\emph{
$($c$)$ For every $r\in\mathbb{R}_+$ and every real number $\epsilon>0$, there is a polynomial $p\in\mathbb{Q}_+[X]$ such that $p(r)\leq\epsilon$ and 
\begin{equation}
\overline{x}+p(\overline{u})\succcurlyeq\overline{y}.
\end{equation}}

\noindent\emph{
$($d$)$ For every $r\in\mathbb{R}_+$ and every real number $\epsilon>0$, there is a polynomial $p\in\mathbb{Q}_+[X]$ such that $p(r)\leq1+\epsilon$ and 
\begin{equation}\label{dd}
p(\overline{u})*\overline{x}\succcurlyeq\overline{y}.
\end{equation}
}

\begin{proof}
The proof of the implication ``$($b$)\Rightarrow($c$)$" is almost the same as  its counterpart in the proof of Theorem 3.1 in \cite{fritz2021generalization}.

The proof of ``$($c$)\Rightarrow($d$)$" is also similar to the corresponding part of the proof of Theorem 3.1 in \cite{fritz2021generalization}. But since we are dealing with noncommutative semialgebras, there are some tiny differences. So we rewrite the proof of ``$($c$)\Rightarrow($d$)$" for clearness: 

Since $x\in S$ and $x\neq0_S$, $x\notin\mathcal{O}$, we have $\overline{x}\neq\overline{0}_S$. By Proposition  \ref{poweruniversal}, $\overline{u}$ is power universal in $F$. Hence there is an integer $k\in\mathbb{N}$ so that $\overline{u}^k*\overline{x}\succcurlyeq\overline{1}_S$. For any $r\in\mathbb{R}_+$ and any real number $\epsilon>0$, set
\[
\epsilon'=\left\{
\begin{array}{rcl}
\frac{\epsilon}{r^k},&\text{if}&r>0,\\
1,&\text{if}&r=0.
\end{array}
\right.
\]
Applying (c) to the numbers $r$ and $\epsilon'$, we obtain $p_0\in\mathbb{Q}_+[X]$ such that $p_0(r)\leq\epsilon'$ and $p_0(\overline{u})+\overline{x}\succcurlyeq\overline{y}$. Set $p=1+p_0X^k$, then \[p(r)=1+p_0(r)r^k\leq1+\epsilon'r^k\leq1+\epsilon\] and 
\[p(\overline{u})*\overline{x}=\overline{x}+p_0(\overline{u})*\overline{u}^k*\overline{x}\succcurlyeq \overline{x}+p_0(\overline{u})\succcurlyeq\overline{y}.\]
Hence (d) is deduced.

``(d)$\Rightarrow$(a)": Let $f: S\rightarrow\mathbb{R}_+$  be a monotone semialgebra homomorphism, we need to show that $f(x)\geq f(y)$. By Remark \ref{realvaluedextension}, the map \[f^F: F\rightarrow \mathbb{R}_+,\overline{a}\mapsto f^W(a) \,\;\;(\forall a\in W)\] is a well-defined monotone semialgebra homomorphism. Setting $r=f^F(\overline{u})$ and applying $f^F$ to inequality (\ref{dd}), we obtain 
\[p(r)f^F(\overline{x})\geq f^F(\overline{y}),\] where $p(r)\leq1+\epsilon$, $f^F(\overline{x})=f^W(x)=f(x)$ and $f^F(\overline{y})=f(y)$. Thus, we have $(1+\epsilon)f(x)\geq f(y)$ for any real number $\epsilon>0$. Taking $\epsilon\rightarrow0$, we obtain $f(x)\geq f(y)$.

``(a)$\Rightarrow$(b)": This part of the proof is almost the same as  the proof of Theorem 3.1 in \cite{fritz2021generalization}. This is mainly due to the well-defined semialgebra of the fractions of a noncommutative semialgebra given in previous  sections. We only sketch  the proof of this part below. 

Assume that $F$ is \emph{order-cancellative}, i.e.  $g+w\succcurlyeq h+w$ implies $g\succcurlyeq h$ ($\forall g,h,w\in F$). Let $V$ be the linear space defined in Section \ref{sec:V}. Define \[C=\{v\in V\,|\,v\unrhd0_{_V}\}=\{g-h\in V\,|\,g,h\in F,g\succcurlyeq h\},\]
then $C$ is a cone in  $V$ by (\ref{Ccone}). For any real number $\epsilon>0$, set 
\[\mathcal{N}_\epsilon=\{a\in F\,|\,\exists m\in\mathbb{N}\text{ \emph{s.t.} }a\preccurlyeq\sum^m_{j=0}\epsilon^{j+1}\cdot\overline{u}^j\},\]
and define $\mathcal{N}_\epsilon-\mathcal{N}_\epsilon=\{a-b\in V\,|\,a,b\in\mathcal{N}_\epsilon\}$. Then, as in \cite{fritz2021generalization}, using the set $\mathcal{N}_\epsilon-\mathcal{N}_\epsilon$ as $\epsilon$ varies a basis of neighborhoods of $0_{_V}$ \hspace{-1mm}makes $V$ a locally convex topological vector space. Denoting by $\overline{C}$ the closure of $C$ in that topology,  we have
\begin{equation}\label{inClosure}
\overline{x}-\overline{y}\in\overline{C}\;\,\Leftrightarrow \text{ the condition (b) holds,}
\end{equation}
where the proof of the direction ``$\Rightarrow$" relies on the assumption that $\preccurlyeq$ is order-cancellative. Denote by $V^*$ the dual space of $V$, \emph{i.e.}, the topological vector space consisting of all continuous linear (real-valued) functions on $V$ and equipped with the weak-* topology. Define
\[C^*=\big\{\ell\in V^*\,|\,\ell(C)\subset[0,\infty)\big\},\]then $C^*$ is a closed convex cone in $V^*$ such that 
\begin{equation}\label{KM}
C^*=\overline{\text{conv}}(\text{Er}(C^*)),
\end{equation}
with Er$(\cdot)$ the set of all the extreme rays of a cone and $\overline{\text{conv}}(\cdot)$ the closure of the convex hull of a set.
In the following, we assume that (b) does not hold and deduce the negation of condition (a).

According to (\ref{inClosure}), ``(b) does not hold" means $\,\overline{x}-\overline{y}\notin\overline{C}$. By the geometric Hahn-Banach theorem \cite{rudin1991functional}, there is an $\ell\in C^*$ such that $\ell(\overline{x}-\overline{y})<0$. Combining with the property in (\ref{KM}), one further claims that 
\begin{equation}\label{exitsf}
\exists f\in\text{Er}(C^*)\subset C^*, ~  f(\overline{x}-\overline{y})<0.
\end{equation}
Then we have  $f\big(\hspace{0.4mm}\overline{1}_S-\overline{0}_S\hspace{-0.4mm}\big)>0$ and the function 
\begin{equation}\label{fhat}
\hat{f}: S\rightarrow\mathbb{R}_+,c\mapsto \frac{1}{f\big(\hspace{0.4mm}\overline{1}_S-\overline{0}_S\hspace{-0.4mm}\big)}f(\overline{c}-\overline{0}_S)
\end{equation}
is a monotone semialgebra homomorphism such that $\hat{f}(x)<\hat{f}(y)$, contradicting  to the condition (a).

If $F$ is not order-cancellative, the method used in \cite{fritz2021generalization} works as well in the present case.
\end{proof}

In  Observations \ref{closure}--\ref{prooffhat},  we verify the condition given in (\ref{inClosure}), the equation (\ref{KM}), the claim given in (\ref{exitsf}) and the claim that  the function $\hat{f}$ in (\ref{fhat}) is a monotone homomorphism. These claims have already been given by Fritz in \cite{fritz2021generalization} and part of the proofs of them are also contained therein. However, we rewrite these proofs for completeness.
Note that the proof of the equation (\ref{KM}) below (Observation \ref{KMobs}) is provided by Fritz through private communications. 

\begin{observation}\label{closure}
$\overline{x}-\overline{y}\in\overline{C}\;\,\Leftrightarrow$ \text{ the condition \emph{(b)} holds}.
\end{observation}
\begin{proof}

The condition $\overline{x}-\overline{y}\in\overline{C}$ implies that
\[\forall \epsilon>0, \,\exists v\in C\text{\emph{ s.t. }}v\in(\overline{x}-\overline{y})+(\mathcal{N}_\epsilon-\mathcal{N}_\epsilon).\]
This means that there are $c,d\in\mathcal{N}_\epsilon$ such that $(\overline{x}-\overline{y})+(c-d)=v\unrhd0_{_V}$. By the assumption that $F$ is order-cancellative, this implies $\overline{x}+c\succcurlyeq\overline{y}+d\succcurlyeq\overline{y}$. Since $c\in\mathcal{N}_\epsilon$, then  (\ref{binq}) holds.

Since the inequality (\ref{binq}) holds, we have 
\[\big(\overline{x}+\sum_{j=0}^m\epsilon^{j+1}\cdot\overline{u}^j\big)-\overline{0}_S\unrhd\overline{y}-\overline{0}_S.\]Adding $\overline{0}_S-\overline{y}$ to the both sides, one obtains
\[\big(\overline{x}+\sum_{j=0}^m\epsilon^{j+1}\cdot\overline{u}^j\big)-\overline{y}\unrhd\overline{y}-\overline{y}=0_{_V},\]
which is equivalent to 
\begin{equation}\label{indicate}
(\overline{x}-\overline{y})+\big(\sum_{j=0}^m\epsilon^{j+1}\cdot\overline{u}^j-\overline{0}_S\big)\in C
\end{equation}
with $\sum_{j=0}^m\epsilon^{j+1}\cdot\overline{u}^j-\overline{0}_S\in\mathcal{N}_\epsilon-\mathcal{N}_\epsilon$. By  inequality (\ref{indicate}), we know that there is a point in $C$ which is also in the neighborhood $(\overline{x}-\overline{y})+(\mathcal{N}_\epsilon-\mathcal{N}_\epsilon)$ for any $\epsilon>0$. Thus $\overline{x}-\overline{y}\in\overline{C}$.
\end{proof}
\begin{observation}\label{KMobs}
$C^*=\overline{\text{\emph{conv}}}(\text{\emph{Er}}(C^*))$.
\end{observation}
\begin{proof}
Proposition 21 in \cite{choquet1962cones} indicates that any \emph{well-capped} \cite{asimow1969extremal} closed convex cone in a locally convex topological vector space is the closed convex hull of its extreme rays. To prove (\ref{KM}), it is sufficient to show that $C^*$ is well-capped (since $V^*$ is already locally convex). In \cite{fritz2021generalization}, Fritz defined the sets
\[C_\epsilon^*=\big\{\ell\in C^*\,|\,\ell(\mathcal{N}_\epsilon-\overline{0}_S\hspace{-0.6mm})\text{ is bounded}\big\}\]
and
\[
D_\epsilon=\big\{\ell\in C_\epsilon^*\,|\,\ell(\mathcal{N}_\epsilon-\overline{0}_S\hspace{-0.6mm})\subset[0,1]\big\},
\]
which satisfy the equation \begin{equation}\label{wellcapped}
C^*=\bigcup_{\epsilon>0}C_\epsilon^*=\bigcup_{\substack{\epsilon>0\\n\in\mathbb{N}}}n D_\epsilon
\end{equation}
and claimed that $D_\epsilon$ was a \emph{cap} \cite{asimow1969extremal} of $C^*$ (hence $C^*$ is well-capped by equation (\ref{wellcapped})). However, the proof of the condition ``the set $C^*\backslash D_\epsilon$ is convex", which is required in the definition of a cap, is omitted in \cite{fritz2021generalization}. The following proof of the condition is provided by Fritz through private communications:

Note that $\sum\limits_{j=0}^m\epsilon^{j+1}\cdot\overline{u}^j\in\mathcal{N}_\epsilon$ for any $m\in\mathbb{N}$, we have
\[
D_\epsilon=\Big\{\ell\in C^*_\epsilon\,\big|\,\forall m\in\mathbb{N}, \ell\big(\sum\limits_{j=0}^m\epsilon^{j+1}\cdot\overline{u}^j-\overline{0}_S\big)\leq1\Big\}.
\]
Then
\[C^*_\epsilon\backslash D_\epsilon=\Big\{\ell\in C^*_\epsilon\,\big|\,\exists m\in\mathbb{N}, \ell\big(\sum\limits_{j=0}^m\epsilon^{j+1}\cdot\overline{u}^j-\overline{0}_S\big)>1\Big\}\]
is convex. Set $f_i\in C^*_\epsilon\backslash D_\epsilon$ with $m_i$ such that $f_i\big(\sum\limits_{j=0}^{m_i}\epsilon^{j+1}\cdot\overline{u}^j-\overline{0}_S\big)>1$, $i=1,2$. For any $\mu_1,\mu_2=1-\mu_1\in(0,1)$ and $\hat{m}=\max\{m_1,m_2\}$, we have
\begin{equation}\label{g1convex}
\left.
\begin{array}{rcl}
&&(\mu_1f_1+\mu_2f_2)\big(\sum\limits_{j=0}^{\hat{m}}\epsilon^{j+1}\cdot\overline{u}^j-\overline{0}_S\big)\\
&=&\mu_1f_1\big(\sum\limits_{j=0}^{\hat{m}}\epsilon^{j+1}\cdot\overline{u}^j-\overline{0}_S\big)+\mu_2f_2\big(\sum\limits_{j=0}^{\hat{m}}\epsilon^{j+1}\cdot\overline{u}^j-\overline{0}_S\big)\\
&\geq&\mu_1f_1\big(\sum\limits_{j=0}^{m_1}\epsilon^{j+1}\cdot\overline{u}^j-\overline{0}_S\big)+\mu_2f_2\big(\sum\limits_{j=0}^{m_2}\epsilon^{j+1}\cdot\overline{u}^j-\overline{0}_S\big)\\
&>&1.
\end{array}
\right.
\end{equation}
Thus $\mu_1f_1+\mu_2f_2\in C_\epsilon^*\backslash D_\epsilon$.

The set $C^*\backslash D_\epsilon$ is also convex.  Set $f_1,f_2\in C^*\backslash D_\epsilon$ and $\mu_1,\mu_2=1-\mu_1\in(0,1)$. If one of $f_1$ and $f_2$ is not in $C^*_\epsilon$ (\emph{i.e.}, is not bounded over the set $\mathcal{N}_\epsilon-\overline{0}_S$), then neither is the functional $\mu_1f_1+\mu_2f_2$. Thus, $\mu_1f_1+\mu_2f_2\in C^*\backslash C^*_\epsilon\subset C^*\backslash D_\epsilon$. On the contrary, if both $f_1$ and $f_2$ are in $C^*_\epsilon$, then they are in $C^*_\epsilon\backslash D_\epsilon$. Now (\ref{g1convex}) implies $\mu_1f_1+\mu_2f_2\in C^*_\epsilon\backslash D_\epsilon\subset C^*\backslash D_\epsilon$.
\end{proof}
\begin{observation}\label{extray}
$\exists f\in\text{\emph{Er}}(C^*)\text{ \emph{s.t.} }f(\overline{x}-\overline{y})<0$.
\end{observation}
\begin{proof}
This part was also omitted in \cite{fritz2021generalization}, hence we give a proof here. Since there is an $\ell\in C^*$ so that $\ell(\overline{x}-\overline{y})<0$, the intersection $\mathcal{H}\cap C^*$ is not empty, where 
\[\mathcal{H}=\{\phi\in V^*\,|\,\phi(\overline{x}-\overline{y})<0\}\]
is an open half-space in $V^*$. It is sufficient to show that there is $f\in\text{Er}(C^*)$ so that $f\in\mathcal{H}$. Assume on the contrary that none of the extreme rays of $C^*$ is in $\mathcal{H}$. By the equation (\ref{KM}), for any $\phi\in C^*$ and any neighborhood $O\ni\phi$, there is $\phi'\in O\cap\text{conv}(\text{Er}(C^*))$. Thus $\phi'=\sum\limits_{i=1}^s\mu_i\phi_i$ for some $s\in\mathbb{N}$, some $\phi_i\in\text{Er}(C^*)$ and some real numbers $\mu_i\in[0,1]$ with $\sum\limits_{i=1}^s\mu_i=1$. Since $\phi_i\notin\mathcal{H}$ by the assumption, $\phi_i(\overline{x}-\overline{y})\geq0$ for any $1\leq i\leq s$. Hence $\phi'(\overline{x}-\overline{y})\geq0$, i.e., $\phi'\notin\mathcal{H}$. This means that $O\cap(V^*\backslash\mathcal{H})\neq\emptyset$ for any neighborhood $O$ of $\phi$. Since $V^*\backslash\mathcal{H}$ is clearly closed, $\phi\in V^*\backslash\mathcal{H}$. Noting that $\phi$ is an arbitrary element in $C^*$, we conclude that $C^*\cap\mathcal{H}=\emptyset$, which is a contradiction.
\end{proof}
\begin{observation}\label{prooffhat}
Let $f\in\text{\emph{Er}}(C^*)$ be any extreme ray of $C^*$, then $f(\overline{1}_S-\overline{0}_S)>0$ and $\hat{f}$ defined  in (\ref{fhat}) is a monotone semialgebra homomorphism.
\end{observation}
\begin{proof}
We may assume, without loos of generality, that for any nonzero $a\in F$, $\overline{1}_S+a\neq\overline{0}_S$. Otherwise, there would be a nonzero $a\in F$ so that $\overline{0}_S=\overline{1}_S+a\succcurlyeq\overline{1}_S$. Note that $\overline{1}_S\succcurlyeq\overline{0}_S$, one concludes that the preordder $\preccurlyeq$ is trivial, as is shown in (\ref{tri}). Now the conditions (a)--(d) are all trivially true, so is the conclusion of the theorem.

Set $f_F(w)=f(w-\overline{0}_S)$ for any $w\in F$, then we have
\[f_F\big((\overline{1}_S+a)^{-1}*w\big)+f_F\big(a*(\overline{1}_S+a)^{-1}*w\big)=f_F(w)\]
by the linearity of $f$. Let
\[
\left\{
\begin{array}{rcl}
g(w)&=&f_F\big((\overline{1}_S+a)^{-1}*w\big),\\
h(w)&=&f_F\big(a*(\overline{1}_S+a)^{-1}*w\big).
\end{array}
\right.
\]
then $g,h$ are both linear and monotone on $F$ since $f_F$ is. Set $g^V(b-d)=g(b)-g(d)$ and $h^V(b-d)=h(b)-h(d)$ for any $b,d\in F$ (or equivalently, for any $b-d\in V$), then both $g^V$ and $h^V$ are well-defined on $V$. Moreover, a routine check indicates that they are monotone and linear maps from $Ve$ to $\mathbb{R}$ and \begin{equation}\label{decom}g^V+h^V=f.\end{equation}

For any nonzero $t\in F$, the operator
\begin{equation}\label{leftmul}
t_*: b-d\mapsto t*b-t*d,\,V\rightarrow V
\end{equation}
is well-defined and continuous (Lemma \ref{t*cont}). Noting that $g^V=f\circ\big((\overline{1}_S+a)^{-1}\big)_*$ and $h^V=f\circ\big(a*(\overline{1}_S+a)^{-1}\big)_*$, one observes that both $g^V$ and $h^V$ are continuous since $f$ is. To conclude, we have $g^V,h^V\in C^*$.

By the equation (\ref{decom}) and the fact that $f$ is an extreme ray, we conclude that there are nonnegative real numbers $s_1$ and $s_2$ such that $g^V=s_1f$ and $h^V=s_2f$. As a ray of $C^*$, $f\neq0$. There is some $b_0-d_0\in V$ such that $f(b_0-d_0)\neq0$. Thus \[g^V\big((\overline{1}_S+a)*b_0-(\overline{1}_S+a)*d_0\big)=f(b_0-d_0)\neq0\] and \[f^V\big((\overline{1}_S+a)*a^{-1}*b_0-(\overline{1}_S+a)*a^{-1}*d_0\big)=f(b_0-d_0)\neq0,\] which means $g^V\neq0$ and $h^V\neq0$. Hence $s_1,s_2\in\mathbb{R}_{>0}$. Therefore, we have $h^V=\gamma g^V$ for the positive number $\gamma=s_2/s_1$ and $h^V(w-\overline{0}_S)=\gamma g^V(w-\overline{0}_S)$ for any $w\in F$, \emph{i.e.},
\begin{equation}\label{fidentity}
f_F\big(a*(\overline{1}_S+a)^{-1}*w\big)=\gamma f_F\big((\overline{1}_S+a)^{-1}*w\big).
\end{equation}
Replacing $(\overline{1}_S+a)^{-1}*w$ by $w$ in (\ref{fidentity}) one obtains 
\begin{equation}\label{fidentitySIM}
f_F(a*w)=\gamma f_F(w)
\end{equation}
for any $w\in F$ (with $\gamma$ in (\ref{fidentitySIM}) depending on $a$). Taking $w=\overline{1}_S$ in (\ref{fidentitySIM}), we have
\begin{equation}\label{gamma}
f_F(a)=\gamma f_F(\overline{1}_S).
\end{equation}
Now $0\neq f(b_0-d_0)=f_F(b_0)-f_F(d_0)$ indicates that $f_F(b_0)\neq0$ or $f_F(d_0)\neq0$. We may assume $f_F(b_0)\neq0$, then clearly $b_0\neq\overline{0}_S$. Taking $a=b_0$ in (\ref{gamma}), we have $f_F(b_0)=\gamma f_F(\overline{1}_S)$. Since $f_F(b_0)\neq0$, we obtain $f_F(\overline{1}_S)\neq0$. In fact, $f_F(\overline{1}_S)>0$ since $f\in C^*$. From (\ref{gamma}) we see that $\gamma=f_F(a)/f_F(\overline{1}_S)$. Substituting that back to (\ref{fidentitySIM}), we have
\begin{equation}\label{nearhomo}
f_F(a*w)=\frac{f_F(a)}{f_F(\overline{1}_S)}f_F(w)
\end{equation}
for any $w\in F$ and any nonzero $a\in F$. Indeed, when $a=\overline{0}_S$, the identity (\ref{nearhomo}) also holds: both of its two sides equal $0$. Define $\hat{f}(c)=\frac{1}{f_F(\overline{1}_S)}f_F(\overline{c})$ for any $c\in S$, then the identity (\ref{nearhomo}) indicates that $\hat{f}: S\rightarrow\mathbb{R}_+$ preserves the multiplications. The facts that it also preserves the additions and the scalar multiplications, and that it is monotone are straightforward. Thus, $\hat{f}$ is a monotone semialgebra homomorphism.
\end{proof}
The following lemma (also claimed in \cite{fritz2021generalization}) is needed in the proof above. 
\begin{lemma}\label{t*cont}
For any nonzero $t\in F$, the map given in (\ref{leftmul}) is well-defined and continuous \emph{w.r.t.} the topology in $V$.
\end{lemma}
\begin{proof}
Suppose that $b-d=b'-d'$ in $V$ for some $b,d,b',d'$ in $F$. Then there is $e\in F$ such that $b+d'+e=b'+d+e$ in $F$. Hence
\[
t*b+t*d'+t*e=t*b'+t*d+t*e,
\]
which means $t*b-t*d=t*b'-t*d'$ in $V$. Therefore, the map $t_*$ is well-defined.

The map $t_*$ is clearly linear. To prove the continuity, it is sufficient to show that it is continuous at the point $0_V=\overline{0}_S-\overline{0}_S$: Let $k\in\mathbb{N}$ such that $\overline{u}^k\succcurlyeq t$. For any $\epsilon>0$, we set $\delta=(\min\{\epsilon,1\})^{k+1}$. Then, for any $b\in\mathcal{N}_\delta$, $b\preccurlyeq\sum\limits_{j=0}^m\delta^{j+1}\cdot\overline{u}^j$ for some $m\in\mathbb{N}$. On the other hand,
\[
\left.
\begin{array}{rcl}
t*b&\preccurlyeq&\overline{u}^k*(\sum\limits_{j=0}^m\delta^{j+1}\cdot\overline{u}^j)\\
&=&\sum\limits_{j=0}^m\delta^{j+1}\cdot\overline{u}^{j+k}\\
&\preccurlyeq&\sum\limits_{j=0}^m(\min\{\epsilon,1\})^{k+j+1}\cdot\overline{u}^{j+k}\\
&\preccurlyeq&\sum\limits_{j=0}^{m+k}(\min\{\epsilon,1\})^{j+1}\cdot\overline{u}^j,
\end{array}
\right.
\]
which means $t*b\in\mathcal{N}_{\min\{\epsilon,1\}}\subset\mathcal{N}_\epsilon$. Therefore, for any $b,d\in\mathcal{N}_\delta$, we have $t*b-t*d\in\mathcal{N}_\epsilon-\mathcal{N}_\epsilon$. That is, $t_*$ is continuous at $0_V$.
\end{proof}
\begin{remark}
Theorem \ref{main} is a noncommutative counterpart of Theorem 2.12 in \cite{fritz2021generalization}, although in the noncommutative case, an equivalent characterization for the preorder relation $\preccurlyeq$ via the original relation $\leq$ similar to the one in (\ref{compare}) is presently not available. 
\end{remark}
\section{Conclusion}
We provide a noncommutative version of the Vergleichsstellensatz proved by Fritz in \cite{fritz2021generalization}. This is the first Vergleichsstellensatz for noncommutative semialgebras in the literature. It characterizes the relaxed preorder induced by monotone homomorphisms to $\mathbb{R}_+$ by ``asymptotic" inequalities in the semialgebra of the fractions. Finding out how it can be applied to other areas such as probability theory and quantum information theory or how it is related to classical noncommutative real algebraic geometry would be interesting topics in the future.

As a byproduct, we provided a particular method to define the semialgebra of the fractions of a noncommutative semialgebra which generalizes the corresponding concept in the commutative case. One can also apply this method to noncommutative semirings by  getting rid of those expressions containing (formal or non-formal) scalar multiplications in the definition.
\section*{Acknowledgement}
We thank Tobias Fritz for the useful discussion and suggestions, which have helped the authors improve the paper a lot. 

This research is supported by the National Key Research Project of China 2018YFA0306702 and the National Natural Science Foundation of China 12071467. 
\bibliographystyle{elsarticle-num}
\bibliography{main}
\begin{appendices}
\section{Two lemmas needed in the proof of Proposition \ref{0class}}
The first lemma below characterizes when an expression in $Q$ lies in $\mathcal{O}$.

\begin{lemma}\label{lemmaQ}
If $w$ is an expression in $Q$, then $w\in\mathcal{O}$ if and only if  $w_S=0_S$.
\end{lemma}
\begin{proof}
``If": We prove inductively on the number $\ell$ of operations that occur in the expression $w$. When $\ell=0$, $w\in W_0=S$. Thus $w=w_S$. But $w_S=0_S$, so $w=0_S\in\mathcal{O}$. Assume that for any $w\in Q$ with the number of operations less than $k$ ($k\geq1$), it holds that $w_S=0_S$ implies $w\in\mathcal{O}$. We then prove that this also holds for those expressions $w\in Q$ containing $k$ operations. Since  $w\in Q$ does not contain any formal inverse, by the definition of $W_{i+1}$ in equation (\ref{W}), $w$ is of the form $w^{(1)}\oplus w^{(2)},w^{(1)}\circledast w^{(2)}$ or $r\odot w^{(1)}$ with $w^{(1)},w^{(2)}\in W$ and $r\in\mathbb{R}_+$. These cases are discussed separately:

If $w=w^{(1)}\oplus w^{(2)}$, then the number of operations contained in $w^{(j)}\,(j=1,2)$ is less than $k$. Moreover, $w^{(1)},w^{(2)}\in Q$ and $w_S=w^{(1)}_S+w^{(2)}_S$. Since $w_S=0_S$, by the assumption that the semialgebra $S$ is zero-sum-free, we see that $w^{(1)}_S=w^{(2)}_S=0_S$. By the inductive assumption, we have $w^{(1)},w^{(2)}\in\mathcal{O}$. Suppose that $w^{(1)}\in\mathcal{O}_{m}$ and $w^{(2)}\in\mathcal{O}_n$, then $w^{(1)},w^{(2)}\in\mathcal{O}_{\max\{m,n\}}$. Thus we have $w=w^{(1)}\oplus w^{(2)}\in\mathcal{O}_{\max\{m,n\}+1}\subset\mathcal{O}$.

The case where $w=w^{(1)}\circledast w^{(2)}$ is similar to the above one. Except that in this case we have $w_S=w^{(1)}_S*w^{(2)}_S$ which implies $w^{(1)}_S=0_S$ or $w^{(2)}_S=0_S$ via the assumption that $S$ is zero-divisor-free. Again, by the inductive assumption, $w^{(1)}\in\mathcal{O}$ or $w^{(2)}\in\mathcal{O}$. Hence $w\in\mathcal{O}_{n+1}$ if $w^{(1)}\in\mathcal{O}_n$ or $w^{(2)}\in\mathcal{O}_n$.

Assume that $w=r\odot w^{(1)}$. If $r=0_\mathbb{R}$, then $w\in\mathcal{O}_{n+1}$ whenever $w^{(1)}\in W_n$ for some $n\in\mathbb{N}$. Suppose that $r$ is a positive real number, then $0_S=w_S=r\cdot w^{(1)}_S$ implies that $w^{(1)}_S=0_S$ by the assumption on $S$. Hence the inductive assumption indicatesl that $w^{(1)}\in\mathcal{O}$. Thus $w\in\mathcal{O}_{m+1}$ whenever $w^{(1)}\in\mathcal{O}_m$ for some  $m\in\mathbb{N}$.

``Only If": If $w\in\mathcal{O}_0$, then $w=0_S$. Clearly, $w_S=0_S$. Assume that for any $0\leq i<k$ ($k\geq1$), $w\in\mathcal{O}_i$ implies $w_S=0_S$. We then prove that $w\in\mathcal{O}_k$ also implies $w_S=0_S$: 

If $w\in\mathcal{O}_{k-1}$, then $w_S=0_S$ follows from the assumption.

If $w\in\mathcal{O}_{k-1}\oplus\mathcal{O}_{k-1}$, then $w=w^{(1)}\oplus w^{(1)}$ with $w^{(1)},w^{(2)}\in\mathcal{O}_{k-1}$. Thus $w_S=w^{(1)}_S+w^{(1)}_S=0_S+0_S=0_S$.

If $w\in\mathcal{O}_{k-1}\circledast W_{k-1}$, then $w=w^{(1)}\circledast a$ with $w^{(1)}\in\mathcal{O}_{k-1}$ and $a\in W_{k-1}$. We have $w_S=w^{(1)}_S*a_S=0_S*a_S=0_S$. When $w\in W_{k-1}\circledast\mathcal{O}_{k-1}$, the proof is similar.

It is straightforward to verify the rest of the cases with $w\in\{0_\mathbb{R}\}\odot W_{k-1}$ and with $w\in\mathbb{R}_+\odot\mathcal{O}_{k-1}$.
\end{proof}

The following lemma indicates that the set $\mathcal{O}$ is $R$-saturated.
\begin{lemma}\label{Osat}
Let the sequence $\{R_i\}$ be as in Proposition \ref{prop:anotherDEFofR}. For any $i\in\mathbb{N}$ and any $x\in W$, if there is an expression $y\in \mathcal{O}$ such that $(x,y)$ or $(y,x)$ is in $R_i$, then $x\in\mathcal{O}$.
\end{lemma}
\begin{proof}
We prove inductively on the index $i$. When $i=0$, it is sufficient to show that if $(a,b)\in R_0$, then $a\in\mathcal{O}$ iff $b\in \mathcal{O}$. Since $R_0=(\cup_{j=0}^6A_j)\bigcup\{(a,a)\,|\,a\in W\}$, we need to discuss several cases separately:

We need two obvious observations: i) $0_S\in\mathcal{O}$, ii) $1_S\notin\mathcal{O}$. The latter one is due to the assumption that $1_S\neq0_S$ in the semialgebra $S$ (hence $1_S\notin\{0_S\}=\mathcal{O}_0$) and the fact that $1_S\notin\mathcal{O}_\ell$ for any integer $\ell>0$.

Using the rules for the set $\mathcal{O}$ in Definition \ref{def:WO} several times, one proves ``$a\in\mathcal{O}$ iff $b\in\mathcal{O}$" in each case where the pair $(a,b)$ admits a particular form given by the elements of the set $A_1\cup A_2\cup A_3$. For example, if $(a,b)\in A_1$ and \[(a,b)=((\hat{b}\oplus\hat{c})\circledast \hat{a},(\hat{b}\circledast\hat{a})\oplus(\hat{c}\circledast\hat{a})),\] then $a\in\mathcal{O}$ implies $\hat{a}\in\mathcal{O}$ or $\hat{b}\oplus\hat{c}\in\mathcal{O}$. If $\hat{a}\in\mathcal{O}$, then both $\hat{b}\circledast\hat{a}$ and $\hat{c}\circledast\hat{a}$ are in $\mathcal{O}$. Hence $b=(\hat{b}\circledast\hat{a})\oplus(\hat{c}\circledast\hat{a})\in\mathcal{O}$. If $\hat{b}\oplus\hat{c}\in\mathcal{O}$, then $\hat{b},\hat{c}\in\mathcal{O}$. Thus $\hat{b}\circledast\hat{a},\hat{c}\circledast\hat{a}\in\mathcal{O}$ and also $b\in\mathcal{O}$. Conversely, $b\in\mathcal{O}$ implies $\hat{b}\oplus\hat{a},\hat{c}\oplus\hat{a}\in\mathcal{O}$. So $\hat{a}\in\mathcal{O}$ or $\hat{a}\notin\mathcal{O}$, $\hat{b},\hat{c}\in\mathcal{O}$. If $\hat{a}\in\mathcal{O}$, then $a=(\hat{b}\oplus\hat{c})\circledast\hat{a}\in\mathcal{O}$. If both $\hat{b}$ and $\hat{c}$ are in $\mathcal{O}$, then so are $\hat{b}\circledast\hat{a}$ and $\hat{c}\circledast\hat{a}$. Thus $b=(\hat{b}\circledast\hat{a})\oplus(\hat{c}\circledast\hat{a})\in\mathcal{O}$.

This is another example: set $(a,b)\in A_2$ and $(a,b)=(1_S\circledast \hat{a},\hat{a})$. If $\hat{a}\in\mathcal{O}$, then $1_S\circledast \hat{a}\in\mathcal{O}$. If $1_S\circledast \hat{a}\in\mathcal{O}$, then we have $\hat{a}\in\mathcal{O}$ since $1_S\notin\mathcal{O}$.

The remaining cases may occur when $(a,b)\in A_1\cup A_2\cup A_3$ can be dealt with similarly. Hence we omit the discussion about them.

When $(a,b)\in A_4$, one finds the first rule in v) of Definition \ref{def:WO} useful: For example, if $(a,b)=((\hat{a}\circledast\hat{b})^{-1},\hat{b}^{-1}\circledast\hat{a}^{-1})$ for some $\hat{a},\hat{b}\in W\backslash\mathcal{O}$, then since $a=(\hat{a}\circledast\hat{b})^{-1}\notin\mathcal{O}$, it is sufficient to prove that $\hat{b}^{-1}\circledast\hat{a}^{-1}\notin\mathcal{O}$. Assume on the contrary that $\hat{b}^{-1}\circledast\hat{a}^{-1}\in\mathcal{O}$, then $\hat{b}^{-1}\in\mathcal{O}$ or $\hat{a}^{-1}\in\mathcal{O}$. Both of these contradict that rule. As usual, we omit the discussion about the other cases that may occur when $(a,b)\in A_4$.

If $(a,b)\in A_5$, then $a,b\in Q$ and $a_S=b_S$. By Lemma \ref{lemmaQ}, it is sufficient to show that $a_S=0_S$ if and only if  $b_S=0_S$. But this is clear since $a_S=b_S$.

The cases with $(a,b)\in A_6\cup \{(a,a)\,|\,a\in W\}$ are trivial. Thus the lemma holds for $i=0$.\vspace{1.5mm}

\noindent{}\emph{The inductive step:}

Now assume that it holds for $0\leq i<k$, we show that it also holds for $i=k$: 

For $(x,y)\in R_k$ or $(y,x)\in R_k$, 
we   consider first  the case  $(x,y)\in R_k$. By the definition of $R_k$ (\ref{Ri}), there are several sub-cases: 

If $(x,y)\in R_{k-1}$, then $x\in\mathcal{O}$ by the inductive assumption. 

If $(x,y)\in R_{k-1}\oplus R_{k-1}$, then $x=x_1\oplus x_2$ and $y=y_1\oplus y_2$ with $(x_1,y_1),(x_2,y_2)\in R_{k-1}$. Since $y\in\mathcal{O}$, $y_1,y_2\in\mathcal{O}$. By the inductive assumption, $x_1,x_2\in\mathcal{O}$. Thus $x=x_1\oplus x_2\in\mathcal{O}$. 

If $(x,y)\in R_{k-1}\circledast R_{k-1}$, then $x=x_1\circledast x_2$ and $y=y_1\circledast y_2$ with $(x_1,y_1),(x_2,y_2)\in R_{k-1}$. Hence $y\in\mathcal{O}$ implies $y_1\in\mathcal{O}$ or $y_2\in\mathcal{O}$. By the inductive assumption, $x_1\in\mathcal{O}$ or $x_2\in\mathcal{O}$. Thus $x=x_1\circledast x_2\in\mathcal{O}$. 

If $(x,y)\in\mathbb{R}_+\odot R_{k-1}$, then $x=r\odot x_1$ and $y=r\odot y_1$ with $(x_1,y_1)\in R_{k-1}$. When $r=0_\mathbb{R}$, $x=r\odot x_1\in\mathcal{O}$. When $r$ is a positive real number, $y=r\odot y_1\in\mathcal{O}$ implies $y_1\in\mathcal{O}$, which, by the inductive assumption, indicates that $x_1\in\mathcal{O}$. Again, $x=r\odot x_1\in\mathcal{O}$.

If $(x,y)\in\text{rev}(R_{k-1})$, that is, $(y,x)\in R_{k-1}$, then $x\in\mathcal{O}$ by the inductive assumption.

Finally, if there is $z\in W$ such that both $(x,z)$ and $(z,y)$ are in $R_{k-1}$, then, by applying the inductive assumption to the pair $(z,y)$, we have $z\in\mathcal{O}$. Applying it again to the pair $(x,z)$, we conclude that $x\in\mathcal{O}$.

By now the proof of the ``$(x,y)\in R_k$" case is done. The proof of the ``$(y,x)\in R_k$" case is similar. Hence we omit it.
\end{proof}
\end{appendices}
\begin{appendices}
\section{Comparing definitions of semialgebras of the fractions in noncommutative and commutative cases}\label{coincide}
For a commutative zero-divisor-free semiring $K$, the semiring of its fractions is usually defined to be the quotient
\begin{equation}\label{def:fr}
K^\text{fr}=\big(K\times (K\backslash\{0_K\})\big)/\hspace{-1mm}\stackrel{\,\hspace{-0.2mm}_\centerdot}{\sim}
\end{equation}
with $\stackrel{\,\hspace{-0.2mm}_\centerdot}{\sim}$ an equivalence relation on the set $K\times (K\backslash\{0_K\})$ such that $(a_1,a_2)\stackrel{\,\hspace{-0.2mm}_\centerdot}{\sim}(b_1,b_2)$ if and only if there is $t\in K\backslash\{0_K\}$ such that $a_1*b_2*t=a_2*b_1*t$ in $K$. The elements in the set $K^\text{fr}$, for instance, the equivalent class containing the the pair $(a_1,a_2)\in K\times (K\backslash\{0_K\})$, is denoted by $\frac{a_1}{a_2}$. The addition and multiplication are defined as usual, and an element $\frac{a_1}{a_2}$ is invertible if and only if $a_1\neq0_K$. This definition is equivalent to the one given by \cite{golan2013semirings} Example 11.7. It can be generalized to the semialgebra case: If $K$ is a commutative zero-divisor-free semialgebra, then we can define the semialgebra of its fractions $K^\text{fr}$ as in equation (\ref{def:fr}), with the same equivalence relation $\stackrel{\,\hspace{-0.2mm}_\centerdot}{\sim}$. Using the same notation, we define the (nonnegative) scalar multiplication by $r\cdot\frac{a_1}{a_2}=\frac{r\cdot a_1}{a_2}$. Then one verifies without difficulty that this $K^{\text{fr}}$ is indeed a semialgebra. 

We  prove that when the semialgebra $S$ (satisfying Assumption \ref{assfree}) is commutative, the semialgebra $F$ of its fractions defined in the Subsection \ref{subsec:frac} is isomorphic to the semialgebra $S^\text{fr}=\big(S\times (S\backslash\{0_S\})\big)/\hspace{-1mm}\stackrel{\,\hspace{-0.2mm}_\centerdot}{\sim}$. 

We need the following result in the proof of Proposition  \ref{isom}.  

\begin{proposition}\label{c2c}
If $S$ is commutative, so is the semialgebra $F$ of its fractions.
\end{proposition}
\begin{proof}
It is sufficient to show that for any $a,b\in W$, $(a\circledast b,b\circledast a)\in R$. We prove that inductively on the number of operations in the expression $a\circledast b$.

If $a\circledast b$ contains only one operation, \emph{i.e.}, $a,b\in W_0=S$, then both $a\circledast b$ and $b\circledast a$ are in $Q$ and $(a\circledast b)_S=a_S*b_S=b_S*a_S=(b\circledast a)_S$. Thus $(a\circledast b,b\circledast a)\in A_5\subset R_0\subset R$.

Now we assume that $(c\circledast d,d\circledast c)\in R$ for any expression $c\circledast d$ which  contains fewer operations than $a\circledast b$ does, and then prove that $(a\circledast b,b\circledast a)\in R$.

We may assume that $a\circledast b$ contains at least two operations. That means $a$ or $b$ contains at least one operation. Without generality loss, we further assume that $a$ contains at least one operation. Hence $a=x\oplus y,x\circledast y,r\odot x$ or $z^{-1}$ for some $x,y\in W,z\in W\backslash\mathcal{O}$ and $r\in\mathbb{R}_+$. 

If $a=x\oplus y$, then $a\circledast b=(x\oplus y)\circledast b\stackrel{\,\hspace{-0.3mm}_R}{\sim} (x\circledast b)\oplus(y\circledast b)\stackrel{\,\hspace{-0.3mm}_R}{\sim}(b\circledast x)\oplus(b\circledast y)\stackrel{\,\hspace{-0.3mm}_R}{\sim} b\circledast(x\oplus y)\stackrel{\,\hspace{-0.3mm}_R}{\sim} b\circledast a$. The second ``$\stackrel{\,\hspace{-0.3mm}_R}{\sim}$" is due to the inductive assumption (note that both $x\circledast b$ and $y\circledast b$ contain fewer operations than $a\circledast b$ does), while the other ones are according to the definition of $R$.

Similarly, if $a=x\circledast y$, then $a\circledast b=(x\circledast y)\circledast b\stackrel{\,\hspace{-0.3mm}_R}{\sim} x\circledast (y\circledast b)\stackrel{\,\hspace{-0.3mm}_R}{\sim} x\circledast (b\circledast y)\stackrel{\,\hspace{-0.3mm}_R}{\sim}(x\circledast b)\circledast y\stackrel{\,\hspace{-0.3mm}_R}{\sim}(b\circledast x)\circledast y\stackrel{\,\hspace{-0.3mm}_R}{\sim} b\circledast(x\circledast y)=b\circledast a$. If $a=r\odot x$, then $a\circledast b=(r\odot x)\circledast b\stackrel{\,\hspace{-0.3mm}_R}{\sim} r\odot(x\circledast b)\stackrel{\,\hspace{-0.3mm}_R}{\sim} r\odot(b\circledast x)\stackrel{\,\hspace{-0.3mm}_R}{\sim} b\circledast(r\odot x)=b\circledast a$. If $a=z^{-1}$ for some $z\in W\backslash\mathcal{O}$, then $z\circledast b$ has less operations than $a\circledast b$ does. Hence $z\circledast b\stackrel{\,\hspace{-0.3mm}_R}{\sim} b\circledast z$, namely $\overline{z}*\overline{b}=\overline{b}*\overline{z}$. Thus $\overline{z}^{-1}*\overline{z}*\overline{b}*\overline{z}^{-1}=\overline{z}^{-1}*\overline{b}*\overline{z}*\overline{z}^{-1}$, which is  $\overline{b}*\overline{z}^{-1}=\overline{z}^{-1}*\overline{b}$. This means $(z^{-1}\circledast b,b\circledast z^{-1})\in R$. And the inductive proof is completed.
\end{proof}

The following lemma defines a map (which is key to the proof of Proposition  \ref{isom}) 
from $W$ to $S^\text{fr}$ when $S$ is commutative:
\begin{lemma}
The recursive assignment below for $\mathscr{G}$ makes it a well-defined map from $W$ to $S^\text{\emph{fr}}$:

1) For any $x\in W_0=S$, set $\mathscr{G}(x)=\frac{x}{1_S}\in S^\text{\emph{fr}}$;

2) Suppose that $x\in W_\ell$ for a minimal $\ell\in\mathbb{N}$ and $\ell>0$, then set
\[
\mathscr{G}(x)=\left\{
\begin{array}{rcl}
\mathscr{G}(a)+\mathscr{G}(b),&if&x=a\oplus b\in W_{\ell-1}\oplus W_{\ell-1},\\\mathscr{G}(a)*\mathscr{G}(b),&if&x=a\circledast b\in W_{\ell-1}\circledast W_{\ell-1},\\
r\cdot \mathscr{G}(a),&if&x=r\odot a\in \mathbb{R}_+\odot W_{\ell-1},\\
(\mathscr{G}(a))^{-1},&if&x=a^{-1}\in (W_{\ell-1}\backslash\mathcal{O}_{\ell-1})^{-1},
\end{array}
\right.
\]
with ``$+$", ``$*$", ``$\cdot$" and ``$\,^{-1}$" the operations in $S^\text{\emph{fr}}$.
\end{lemma}
\begin{proof}
It is sufficient to prove (inductively on the index $i$) the following statement: for any $i\in\mathbb{N}$ and $x\in W_i$, $\mathscr{G}(x)$ is well-defined, and $\mathscr{G}(x)$ is invertible in $S^\text{fr}$ whenever $x\not\in\mathcal{O}_i$.

The ``$i=0$" case is true. We then assume that this statement holds for $i=k-1$ and prove it for $i=k$. Now that $x\in W_k$, there are several sub-cases: 

If $x\in W_{k-1}$, then $\mathscr{G}(x)$ is well-defined. When $x\not\in\mathcal{O}_k$, we also have $x\not\in\mathcal{O}_{k-1}$. Hence $\mathscr{G}(x)$ is invertible.

If $x\not\in W_{k-1}$ but $x=a\oplus b\in W_{k-1}\oplus W_{k-1}$, then $k$ is the minimal nature number such that $x\in W_k$. Hence $\mathscr{G}(x)=\mathscr{G}(a)+\mathscr{G}(b)$ for some $a,b\in W_{k-1}$ and $\mathscr{G}(x)$ is well-defined. When $x=a\oplus b\notin\mathcal{O}_k$, at least one of $a$ and $b$ is not in $\mathcal{O}_{k-1}$. Hence $\mathscr{G}(a)=\frac{a_1}{a_2}$ and $\mathscr{G}(b)=\frac{b_1}{b_2}$ with $a_2,b_2$ nonzero and one of $a_1$ and $b_1$ nonzero. Thus, $a_1*b_2+a_2*b_1\neq0_S$ follows the assumption that $S$ is zero-divisor and zero-sum-free. Hence $\mathscr{G}(x)=\frac{a_1*b_2+a_2*b_1}{a_2*b_2}$ is invertible.

If $x\not\in W_{k-1}$ but $x=a\circledast b\in W_{k-1}\circledast W_{k-1}$, then $\mathscr{G}(x)=\mathscr{G}(a)*\mathscr{G}(b)$ is well-defined. When $x=a\circledast b\notin\mathcal{O}_k$, neither of $a$ and $b$ is in $\mathcal{O}_{k-1}$. Hence,  $\mathscr{G}(a)=\frac{a_1}{a_2}$ and $\mathscr{G}(b)=\frac{b_1}{b_2}$ with none of $a_1,a_2,b_1$ and $b_2$ equals to $0_S$. Thus $a_1*b_1\neq0_S$ and $\mathscr{G}(x)=\frac{a_1*b_1}{a_2*b_2}$ is invertible.

If $x\notin W_{k-1}$ but $x=r\odot a\in\mathbb{R}_+\odot W_{k-1}$, then $\mathscr{G}(x)=r\cdot\mathscr{G}(a)$ is well-defined. When $x=r\odot a\notin\mathcal{O}_k$, $r\neq0$ and $a\notin\mathcal{O}_{k-1}$. Hence $\mathscr{G}(a)=\frac{a_1}{a_2}$ with $a_1\neq0_S$, and $\mathscr{G}(x)=\frac{r\cdot a_1}{a_2}$ is invertible.

If $x\notin W_{k-1}$ but $x=a^{-1}\in (W_{k-1}\backslash\mathcal{O}_{k-1})^{-1}$, then $\mathscr{G}(a)$ is well-defined and invertible. Thus $\mathscr{G}(x)=(\mathscr{G}(a))^{-1}$ is also well-defined. Clearly $\mathscr{G}(x)$ is invertible.
\end{proof}
The lemma below defines a push-forward of the map $\mathscr{G}$ from $W$ to the semialgebra $F$ of the fractions of $S$:
\begin{lemma}Let the sequence $\{R_i\}$ be as in Proposition \ref{prop:anotherDEFofR}. Set $x,y\in W$. If $(x,y)\in R$, then $\mathscr{G}(x)=\mathscr{G}(y)$. That is, $\overline{x}=\overline{y}$ implies $\mathscr{G}(x)=\mathscr{G}(y)$.
\end{lemma}
\begin{proof}
For $(x,y)\in R_0$, we prove that $\mathscr{G}(x)=\mathscr{G}(y)$:

From the definition of $\mathscr{G}$ it follows that for any $a,b\in W$ and $c\in W\backslash\mathcal{O}$,
\begin{equation}\label{huaG}
\left\{
\begin{array}{rcl}
     \mathscr{G}(a\oplus b)&=&\mathscr{G}(a)+\mathscr{G}(b), \\
     \mathscr{G}(a\circledast b)&=&\mathscr{G}(a)*\mathscr{G}(b),\\
     \mathscr{G}(r\odot a) &=&r\cdot \mathscr{G}(a),\\
    \mathscr{G}(c^{-1}) &=&(\mathscr{G}(c))^{-1}.
\end{array}
\right.
\end{equation}
If $(x,y)\in\cup^4_{i=1}A_i$, then $\mathscr{G}(x)=\mathscr{G}(y)$ follows from the above identities. If $(x,y)\in A_5$, \emph{i.e.}, $x,y\in Q$ and $x_S=y_S$, then $\mathscr{G}(x)=\frac{x_S}{1_S}=\frac{y_S}{1_S}=\mathscr{G}(y)$ (by using the identities above recursively for any $a\in Q$, we see that $\mathscr{G}(a)=\frac{a_S}{1_S}$). If $(x,y)\in \{(a,a)\,|\,a\in W\}$, then $x=y$ and $\mathscr{G}(x)=\mathscr{G}(y)$. If $(x,y)\in A_6$, then $x\in\mathcal{O}$ and $y=0_S$. It is straightforward to prove that for any $a\in\mathcal{O}$, $\mathscr{G}(a)=\frac{0_S}{1_S}$, hence $\mathscr{G}(x)=\mathscr{G}(y)$. Thus, $\mathscr{G}(x)=\mathscr{G}(y)$ for any $(x,y)\in R_0$.

We then assume that $\mathscr{G}(x)=\mathscr{G}(y)$ for any $(x,y)\in R_k$, and prove this identity for $R_{k+1}$:

If $(x,y)\in R_k\cup\,\text{rev}(R_k)$, then $\mathscr{G}(x)=\mathscr{G}(y)$ follows right from the assumption. If $(x,y)=(x_1\oplus x_2,y_1\oplus y_2)\in R_k\oplus R_k$ with $(x_1,y_1)$ and $(x_2,y_2)$ in $R_k$, then $\mathscr{G}(x_j)=\mathscr{G}(y_j)$ for $j=1,2$. Thus $\mathscr{G}(x)=\mathscr{G}(x_1)+\mathscr{G}(x_2)=\mathscr{G}(y_1)+\mathscr{G}(y_2)=\mathscr{G}(y)$. If $(x,y)$ is in the sets $R_k\circledast R_k$, $\mathbb{R}_+\odot R_k$ or $R_k\backslash((\mathcal{O}\times W)\cup(W\times\mathcal{O}))$, the identity follows similarly. If there is $z\in W$ such that $(x,z),(z,y)\in R_k$, then we have $\mathscr{G}(x)=\mathscr{G}(z)=\mathscr{G}(y)$.
\end{proof}
We can thus define the map $
\mathscr{\overline{G}}: F\rightarrow S^{\text{fr}}, \overline{x}\mapsto\mathscr{G}(x)$.

Let $E$ and $K$ be two semialgebras. A homomorphism $f$ from $E$ to $K$ is a \emph{semialgebra isomorphism} (or, simply \emph{isomorphism}) if there is a homomorphism $g$ from $K$ to $E$ such that $g\circ f=id_E$ and $f\circ g=id_K$. If there is an isomorphism from $E$ to $K$, we say they are \emph{isomorphic} to each other. We will see that $\mathscr{\overline{G}}$ is a semialgebra isomorphism:
\begin{proposition}\label{isom}
Define $\mathscr{H}: S^{\text{\emph{fr}}}\rightarrow F,\frac{a}{b}\mapsto\overline{a}*\overline{b}^{-1}$ and let $\mathscr{\overline{G}}$ be as above, then $\mathscr{H}$ and $\overline{\mathscr{G}}$ are semialgebra isomorphisms which are also compatible with the inversions such that $\mathscr{\overline{G}}\circ\mathscr{H}=id_{S^\text{\emph{fr}}}$ and $\mathscr{H}\circ\mathscr{\overline{G}}=id_{F}$.
\end{proposition}
\begin{proof}
The map $\mathscr{H}$ is well-defined: for any $\frac{a}{b}=\frac{c}{d}$ in $S^\text{fr}$, there is nonzero $t\in S$ so that $a*d*t=b*c*t$. Hence $(a\circledast d)\circledast t\stackrel{\,\hspace{-0.3mm}_R}{\sim} (b\circledast c)\circledast t$ and therefore $(\overline{a}*\overline{d})*\overline{t}=(\overline{b}*\overline{c})*\overline{t}$ in $F$. Since none of $b,d,t$ is in $\mathcal{O}$, their corresponding equivalent classes are invertible in $F$. Because $F$ is commutative, we have $\overline{a}*\overline{b}^{-1}=\overline{c}*\overline{d}^{-1}$.

Using, if necessary, the fact that $F$ is commutative, one verifies directly that $\mathscr{H}$ is compatible with the additions, the multiplications, the scalar multiplications, and the inversions in $S^\text{fr}$ and $F$. Moreover, $\mathscr{H}(\frac{0_S}{1_S})=\overline{0}_S$ and $\mathscr{H}(\frac{1_S}{1_S})=\overline{1}_S$. Thus $\mathscr{H}$ is a semialgebra homomorphism compatible with the inversions.

Noting that, for any $a,b\in W$, $\mathscr{\overline{G}}(\overline{a}+\overline{b})=\mathscr{\overline{G}}(\overline{a\oplus b})=\mathscr{G}(a\oplus b)$ and that $\mathscr{G}(a\oplus b)=\mathscr{G}(a)+\mathscr{G}(b)=\mathscr{\overline{G}}(\overline{a})+\mathscr{\overline{G}}(\overline{b})$ by equations (\ref{huaG}), we see that $\mathscr{\overline{G}}$ is compatible with the additions in $F$ and $S^{\text{fr}}$. Similarly, we conclude that $\mathscr{\overline{G}}$ is also compatible with the other three pairs of operations in $F$ and $S^{\text{fr}}$. Moreover, $\mathscr{\overline{G}}(\overline{0}_S)=\mathscr{G}(0_S)=\frac{0_S}{1_S}$ and $\mathscr{\overline{G}}(\overline{1}_S)=\mathscr{G}(1_S)=\frac{1_S}{1_S}$. Hence $\mathscr{\overline{G}}$ is a semialgebra homomorphism compatible with the inversions.

Finally, since $\mathscr{\overline{G}}(\mathscr{H}(\frac{a}{b}))=\mathscr{\overline{G}}(\overline{a}*\overline{b}^{-1})=\mathscr{\overline{G}}(\overline{a\circledast b^{-1}})=\mathscr{G}(a\circledast b^{-1})=\mathscr{G}(a)*(\mathscr{G}(b))^{-1}=\frac{a}{1_S}*\frac{1_S}{b}=\frac{a}{b}$ for any $(a,b)\in S\times(S\backslash\{0_S\})$, we have $\overline{\mathscr{G}}\circ\mathscr{H}=\emph{id}_{S^\text{fr}}$. To prove $\mathscr{H}\circ\overline{\mathscr{G}}=\emph{id}_F$, it is sufficient to show that for any $i\in\mathbb{N}$ and any $w\in W_i$, it holds that $\mathscr{H}(\mathscr{G}(w))=\overline{w}$. This can be done inductively on the index $i$: When $i=0$, $w\in S$. Thus $\mathscr{H}(\mathscr{G}(w))=\mathscr{H}(\frac{w}{1_S})=\overline{w}*\overline{1}_S^{-1}=\overline{w}$. We suppose that this claim holds for $i=k$ and then prove it for $i=k+1$. 

If $w=a\oplus b\in W_k\oplus W_k$, with $a,b\in W_k$, and set $\mathscr{G}(a)=\frac{a_1}{a_2}$ and $\mathscr{G}(b)=\frac{b_1}{b_2}$ for some $(a_1,a_2),(b_1,b_2)\in S\times(S\backslash\{0_S\})$, then we have $\mathscr{H}(\mathscr{G}(w))=\mathscr{H}(\frac{a_1}{a_2}+\frac{b_1}{b_2})=\overline{a_1*b_2+a_2*b_1}*\overline{a_2*b_2}^{-1}=\overline{a_1}*\overline{a_2}^{-1}+\overline{b_1}*\overline{b_2}^{-1}=\mathscr{H}(\frac{a_1}{a_2})+\mathscr{H}(\frac{b_1}{b_2})=\mathscr{H}(\mathscr{G}(a))+\mathscr{H}(\mathscr{G}(b))=\overline{a}+\overline{b}=\overline{w}$, where the second last equality is due to the inductive assumption. The rest cases where $w\in W_k\circledast W_i, \mathbb{R}_+\odot W_k$ or $(W_k\backslash\mathcal{O}_k)^{-1}$ can be dealt with in a similar way. Hence $\mathscr{H}(\mathscr{G}(w))=\overline{w}$ for any $w\in W$ and therefore we have $\mathscr{H}\circ\mathscr{\overline{G}}=id_{F}$. Thus both $\mathscr{H}$ and $\overline{\mathscr{G}}$ are semialgebra isomorphisms.
\end{proof}
\end{appendices}
\begin{appendices}
\section{The derived preorder relations in the commutative and the noncommutative cases coincide}\label{fracorder}
The following proposition characterizes the derived preorder relation in the commutative case (see condition (\ref{compare})) recursively:
\begin{proposition}\label{orderrelation}
Let $K$ be a commutative zero-divisor-free semialgebra with preorder relation ``$\leq$" $($such that $0_K\leq1_K)$ and the derived preorder relation ``$\preccurlyeq$" on the semialgebra $K^\text{\emph{fr}}$ of its fractions as in (\ref{compare}), then for any $g,h\in K^\text{\emph{fr}}$, $g\preccurlyeq h$ iff the pair $(g,h)$ is in the following set $\Gamma=\cup_{i=0}^\infty\Gamma_i$ with
\[
\left\{
\begin{array}{lcl}
\Gamma_0&=&\{(\frac{x}{1_K},\frac{y}{1_K}),(\frac{0_K}{1_K},a),(a,a)\,|\,x\leq y,x,y\in K, a\in K^{\text{\emph{fr}}}\},\\
\Gamma_{i+1}&=&\Gamma_i\cup(\mathbb{R}_+\cdot \Gamma_i)\cup(\Gamma_i+\Gamma_i)\cup( K^{\text{\emph{fr}}}*\Gamma_i)\\
&&\cup\{(a,c)\,|\,\exists b\in K^{\text{\emph{fr}}}\text{ such that }(a,b),(b,c)\in \Gamma_i\}.
\end{array}
\right.
\]
\end{proposition}
\begin{proof}
``If": It is clear that very pair $(g,h)$ in $\Gamma_0$ satisfies $g\preccurlyeq h$ by the condition (\ref{compare}). Noting that the preorder relation defined in condition (\ref{compare}) satisfies the implications in (\ref{preoc}) and (\ref{scalar}) (since the preorder relation ``$\leq$" does), that it is transitive and that every pair $(g,h)\in\Gamma_0$ satisfies $g\preccurlyeq h$, we see that every pair $(g,h)\in\Gamma_1$ also satisfies $g\preccurlyeq h$. By induction on the index $i$ of $\Gamma_i$, we see that every pair $(g,h)\in\Gamma$ also satisfies $g\preccurlyeq h$.

``Only if": Suppose that $g=\frac{x}{a}\preccurlyeq h=\frac{y}{b}$ for some $x,y\in K$ and some $a,b\in K\backslash\{0_K\}$, we need to show that $(\frac{x}{a},\frac{y}{b})\in\Gamma$. Since $\frac{x}{a}\preccurlyeq\frac{y}{b}$, by condition (\ref{compare}) there is $t\in K\backslash\{0_K\}$ such that $x*b*t\leq y*a*t$ in $K$. Hence $(\frac{x*b*t}{1_K},\frac{y*a*t}{1_K})\in\Gamma_0$, thus $(\frac{x}{a},\frac{y}{b})=\frac{1_K}{a*b*t}*(\frac{x*b*t}{1_K},\frac{y*a*t}{1_K})\in\Gamma_1\subset\Gamma$.
\end{proof}
Note that the recursive definition of $\Gamma$ above is of the same form as the equations in (\ref{Li}) once we identify the semialgebra $F$ of the fractions of $S$ defined in Subsection \ref{subsec:frac} and the semialgebra $S^\text{fr}$ defined in Appendix \ref{coincide} through the bijections $\overline{\mathscr{G}}$ and $\mathscr{H}$ in Proposition \ref{isom}. In other words, we have $\overline{\mathscr{G}}(L)=\Gamma$ and $\mathscr{H}(\Gamma)=L$, that means
\begin{proposition}\label{preorderisom}
Let $(S,\leq)$ be a commutative preordered semialgebra satisfying Assumption \ref{assfree} such that $1_S\geq0_S$ and $(F,\preccurlyeq)$ be the semialgebra of the fractions of $S$ defined in Subsection \ref{subsec:frac} with $\preccurlyeq$ the derived preorder relation defined in Section \ref{sec:derivedPreorder}. Let $S^\text{\emph{fr}}$ be the semialgebra of the fractions of $S$ defined in Appendix \ref{coincide} with a preorder relation $\preccurlyeq$ given by condition (\ref{compare}). Then the map $\overline{\mathscr{G}}$ and $\mathscr{H}$ are \emph{preordered semialgebra isomorphisms}. That means, we have $g\preccurlyeq h$ in $F$ if and only if $\mathscr{\overline{G}}(g)\preccurlyeq\mathscr{\overline{G}}(h)$ in $S^\text{\emph{fr}}$ for any $g,h\in F$, or equivalently, we have $a\preccurlyeq b$ in $S^\text{\emph{fr}}$ if and only if $\mathscr{H}(a)\preccurlyeq\mathscr{H}(b)$ for any $a,b\in S^\text{\emph{fr}}$.
\end{proposition}
\end{appendices}
\end{document}